\documentclass[12pt, twoside]{article}
\usepackage{amsmath,amssymb}
\usepackage{amsthm}
\usepackage{graphicx}
\usepackage{epsfig}
\usepackage{cite}
\usepackage{subfigure}
\usepackage{hyperref}

\newtheorem{theorem}{Theorem}[section]
\newtheorem{lemma}{Lemma}[section]

\newtheorem{assumption}{Assumption}[section]

\numberwithin{equation}{section}

\numberwithin{equation}{section} \makeatletter
\begin{document}

\title{\textbf{Bifurcation analysis of an age structured HIV infection model
with both virus-to-cell and cell-to-cell transmissions}}
\author{Xiangming Zhang{\small \textsc{$^{a,}$\thanks{%
Research was partially supported by NSFC (Grant No. 11471044 and 11771044) and the
Fundamental Research Funds for the Central Universities.}}} and Zhihua Liu%
{\small \textsc{$^{a,*,}$}}\textsc{\thanks{%
{\small Corresponding author. \newline \indent~~E-mail addresses: xiangmingzhang@mail.bnu.edu.cn (X. Zhang), zhihualiu@bnu.edu.cn (Z. Liu).}}} \\
$^{a}$School of Mathematical Sciences, Beijing Normal University,\\
Beijing, 100875, People's Republic of China}

\date{}
\maketitle
\begin{abstract}
We make a mathematical analysis of an age structured HIV infection model
with both virus-to-cell and cell-to-cell transmissions to understand the
dynamical behavior of HIV infection in vivo. In the model, we consider the
proliferation of uninfected CD$4^{+}$ T cells by a logistic function and the
infected CD$4^{+}$ T cells are assumed to have an infection-age structure.
Our main results concern the Hopf bifurcation of the model by using the
theory of integrated semigroup and the Hopf bifurcation theory for
semilinear equations with non-dense domain. Bifurcation analysis indicates
that there exist some parameter values such that this HIV infection model
has a non-trivial periodic solution which bifurcates from the positive
equilibrium. The numerical simulations are also carried out.

\textbf{Key words:} HIV infection model; Logistic growth; Virus-to-cell; Cell-to-cell; Age structure; Non-densely defined Cauchy problem; Hopf bifurcation

\textbf{Mathematics Subject Classification:} 34C20; 34K15; 37L10
\end{abstract}

\section{Introduction}

\noindent

The human immunodeficiency virus, HIV, gives rise to acquired immune
deficiency syndrome, AIDS. Nowadays AIDS still severely threats the people's
heath all over the world. The main target of HIV infection is a class of
lymphocytes, or white blood cells, known as CD$4^{+}$ T cells. In general,
there are two fundamental modes of viral infection and transmission, one is
the classical virus-to-cell infection and the other is direct cell-to-cell
transmission. In the classical mode, viral particles that released from
infected cells arbitrarily move around any distance to discover a new target
cell to infect. For the direct cell-to-cell transmission, HIV infection can
occur by the movement of viruses by means of direct contact between infected
cells and uninfected cells via some structures, such as membrane nanotubes
\cite{Sattentau-COV-2011}. In recent years, HIV infection model which
involves different infection modes, such as the classical virus-to-cell
infection \cite%
{PerelsonNelson-SIAMRev-1999,RJDeBoerASPerelson-JTB-1998,ZhouSongShi-AMC-2008,SongZhouZhao-AMM-2010,Browne-NARWA-2015}%
, the direct cell-to-cell transmission \cite%
{CulshawRuanWebb-JMB-2003,KomarovaWodarz-MB-2013,LaiZou-2015-JMAA}, and both
virus-to-cell infection and cell-to-cell transmission \cite%
{LaiZou-SIAM-JAM-2014,HuHuLiao-MCS-2016,WangLangZou-NARWA-2017}, has been
extensively studied by many scholars.

In population dynamics, age structure, in many situations, can affect
population size and growth in a major way because different ages indicate
different reproduction and survival abilities and different behaviors.
Generally, the progress of disease propagation and individual interactions
are modeled by using an ODE system \cite%
{PerelsonNelson-SIAMRev-1999,LaiZou-2015-JMAA} or DDE system \cite%
{ZhouSongShi-AMC-2008,SongZhouZhao-AMM-2010,HuHuLiao-MCS-2016}. When age
structure is introduced into individual interactions, population dynamical
models become considerably intricate. Recently, as the significance of age
structure in populations has become increasingly prevalent, there has been
explosively growing literature dealing with all kinds of aspects of
interacting populations with age structure \cite%
{Browne-NARWA-2015,MimmoIannelli-1995,WangLangZou-NARWA-2017,LiuLi-JNS-2015,TangLiu-AMM-2016,WangLiu-JMAA-2012, LiuMagalRuan-DCDS-B-2016,LiuTangMagal-DCDS-B-2015,FuLiuMagal-2015}%
. In particular, \cite%
{LiuLi-JNS-2015,TangLiu-AMM-2016,WangLiu-JMAA-2012,LiuMagalRuan-DCDS-B-2016,
LiuTangMagal-DCDS-B-2015,FuLiuMagal-2015} considered the age-structured
model as a non-densely defined Cauchy problem and discussed the existence of
Hopf bifurcation.

\cite{PerelsonNelson-SIAMRev-1999} considered the proliferation of
uninfected $T$ cells by a logistic function and formulated the following
model
\begin{equation*}
\left\{
\begin{array}{ccl}
\frac{dT}{dt} & = & \Lambda+rT\left(1-\frac{T}{K}\right)-\mu T-\beta_{1}VT,
\\
\frac{dI}{dt} & = & \beta_{1}VT-\sigma I, \\
\frac{dV}{dt} & = & N\sigma I-cV, \\
\end{array}
\right.
\end{equation*}
where $T$, $I$ and $V$ denote the population of uninfected CD$4^{+}$ T
cells, productively infected CD$4^{+}$ T cells and virus, respectively. They
treated the logistic proliferation term $rT\left(1-\frac{T+I}{K}\right)$ in
which $r$ is the maximum proliferation rate and $K$ is the uninfected CD$%
4^{+}$ T cells population density at which proliferation shuts off. Since
the proportion of productively infected CD$4^{+}$ T cells $I$ is very small
and thus it is reasonable to ignore this correction. They eventually
represented the proliferation of uninfected CD$4^{+}$ T cells by a logistic
function $rT\left(1-\frac{T}{K}\right)$. However, the authors only consider
the classical virus-to-cell infection and neglect the direct cell-to-cell
transmission.

\cite{WangLangZou-NARWA-2017} incorporated the two modes (virus-to-cell
infection and cell-to-cell transmission) of transmission into a classic
model and considered the following model system
\begin{equation*}
\left\{
\begin{array}{l}
\frac{dT(t)}{dt} = \Lambda-\mu
T(t)-\beta_{1}T(t)V(t)-\beta_{2}T(t)\int_{0}^{+\infty}{q(a)i(t,a)da}, \\
\frac{\partial i(t,a)}{\partial t}+\frac{\partial i(t,a)}{\partial a}%
=-\sigma(a)i(t,a), \\
\frac{dV(t)}{dt} = \int_{0}^{\infty}{p(a)i(t,a)da}-cV(t), \\
\end{array}
\right.
\end{equation*}
with the boundary and initial conditions
\begin{equation*}
\left\{
\begin{array}{l}
i(t,0)= \beta_{1}T(t)V(t)+\beta_{2}T(t)\int_{0}^{+\infty}{q(a)i(t,a)da}, \\
T(0)=x_{0}>0, \quad i(0,a)=i_{0}(a)\in L_{+}^{1}(0,\infty), \quad
V(0)=V_{0}>0, \\
\end{array}
\right.
\end{equation*}
where $T(t)$ and $V(t)$ denote the concentration of uninfected CD$4^{+}$ T
cells and infectious virus at $t$, respectively; $i(t,a)$ denotes the
concentration of infected CD$4^{+}$ T cells of infection age $a$ at time $t$%
, $\Lambda$ is the constant recruitment rate, $\mu$ is the natural death
rate of uninfected CD$4^{+}$ T cells, $\beta_{1}$ is the rate at which an
uninfected CD$4^{+}$ T cell becomes infected by an infectious virus, $c$ is
the clearance rate of virions, $\sigma(a)$ is the death rate of infected CD$%
4^{+}$ T cells related to infection age $a$, and $p(a)$ is the viral
production rate of an infected CD$4^{+}$ T cell with infection age $a$, $q(a)
$ measures variance of the infectivity of infected CD$4^{+}$ T cells with
respect to the infection age $a$. They analysed the relative compactness and
persistence of the solution semiflow and existence of a global attractor and
investigated how the rate functions $p(a)$, $q(a)$, and $\sigma(a)$ affected
the global dynamics. They do not take into account any more dynamical behaviors such as bifurcation behaviors.

Just as described above, few scholars simultaneously considered the logistic
proliferation function of uninfected CD$4^{+}$ T cells and the two
predominant infection modes of HIV in a model. As is known that the age
structure model can be considered as abstract Cauchy problems with non-dense
domain. Inspired by the papers \cite%
{PerelsonNelson-SIAMRev-1999,WangLangZou-NARWA-2017,TangLiu-AMM-2016,WangLiu-JMAA-2012}%
, we attempt to investigate the following HIV infection-age structured model
(\ref{system}) by applying the theory of integrated semigroup and Hopf
bifurcation theory \cite{LiuMagalRuan-ZAMP-2011}. A schematic diagram of the
model (\ref{system}) is shown in Figure \ref{systemfigure} and the dynamics
of such a model can be written as
\begin{figure}[tbp]
\setlength{\belowcaptionskip}{2pt} \centerline{%
\includegraphics[height=3.5cm]{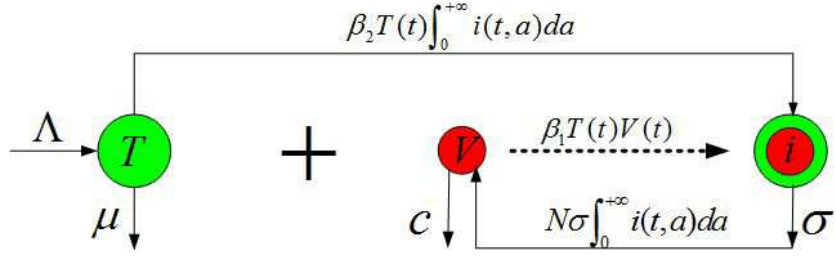}} \renewcommand{\figurename}{%
\footnotesize{Figure}}
\caption{{\protect\footnotesize {Age structured HIV infection model
flowchart.}}}
\label{systemfigure}
\end{figure}
\begin{equation}  \label{system}
\left\{
\begin{array}{l}
\frac{dT(t)}{dt}=\Lambda-\mu T(t)+rT(t)\left(1-\frac{T(t)+\int_{0}^{+\infty}{%
i(t,a)da}}{K}\right)-\beta_{1} T(t)V(t) \\
\quad\quad\quad-\beta_{2} T(t)\int_{0}^{+\infty}{i(t,a)da}, \\
\frac{dV(t)}{dt}=N\sigma\int_{0}^{+\infty}{i(t,a)da}-cV(t), \\
\frac{\partial i(t,a)}{\partial t}+\frac{\partial i(t,a)}{\partial a} =
-\sigma i(t,a), \\
i(t,0)=\beta_{1} T(t)V(t)+\beta_{2} T(t)\int_{0}^{+\infty}{\beta(a)i(t,a)da}%
,\quad t>0, \\
T(0)=T_{0}\geq0,\quad V(0)=V_{0}\geq0,\quad i(0,\cdot)=i_{0}\in L_{+}
^{1}((0, + \infty ),\mathbb{R}),%
\end{array}
\right.
\end{equation}
where $T(t)$ denotes the concentration of uninfected CD$4^{+}$ T cells at
time $t$, $i(t,a)$ denotes the concentration of infected CD$4^{+}$ T cells
of infection age $a$ at time $t$, and $V(t)$ denotes the concentration of
infectious virus at $t$. All parameters of model (\ref{system}) are positive
constants and the parameters description are presented in Table \ref%
{parametersdescription}.
\begin{table}[tbp]
\caption{The parameters description of the HIV model (\protect\ref{system}).}
\label{parametersdescription}
\centering
\doublerulesep=0.4pt
\begin{tabular}{p{2cm}p{13cm}}
\hline
\multicolumn{1}{l}{Parameter} & \multicolumn{1}{l}{Description} \\ \hline
$\Lambda$ & {\footnotesize the rate at which new CD$4^{+}$ T cells are
created from sources within the body.} \\
$r$ & {\footnotesize the maximum proliferation rate of uninfected CD$4^{+}$
T cells.} \\
$K$ & {\footnotesize the CD$4^{+}$ T cells population density at which
proliferation shuts off.} \\
$N$ & {\footnotesize the number of virons produced the infected CD$4^{+}$ T
cells during its lifetime.} \\
$\beta_{1}$ & {\footnotesize the rate at which an uninfected CD$4^{+}$ T
cell becomes infected by an infectious virus.} \\
$\beta_{2}$ & {\footnotesize the infection rate of productively infected CD$%
4^{+}$ T cells.} \\
$\mu$ & {\footnotesize the natural death rate of uninfected CD$4^{+}$ T
cells.} \\
$c$ & {\footnotesize the clearance rate of virions.} \\
$\sigma$ & {\footnotesize the death rate of infected CD$4^{+}$ T cells.} \\
\hline
\end{tabular}%
\end{table}
Throughout the paper, $\beta(a)$ is an age-specific fertility function
related to infection age $a$ and satisfies the following assumption \ref%
{assumption1}.

\begin{assumption}
\label{assumption1} Assume that
\begin{equation*}
\beta(a):=\left\{
\begin{array}{cl}
\beta^{*}, & \quad \mbox{if} \quad a\geq \tau, \\
0, & \quad \mbox{if} \quad a\in (0,\tau), \\
\end{array}
\right.
\end{equation*}
where $\tau>0$ and $\beta^{*}>0$. Moreover, it is reasonable and favorable
for the infected CD$4^{+}$ T cells to show a stable trend to assume that $%
\int_{0}^{+\infty}{\beta(a)e^{-\sigma a}da}=1$, where $e^{-\sigma a}$
denotes the probability for an infected $T$ cell to survive to age $a$.
\end{assumption}

The paper is organized as follows. In Section 2, we reformulate system (\ref%
{system}) as an abstract non-densely defined Cauchy problem and study the
equilibrium, linearized equation and characteristic equation. The existence
of Hopf bifurcation is proved in Section 3. Some numerical simulations and
conclusions are presented in Section 4.

\section{Preliminaries}

\subsection{Rescaling time and age}

\noindent

In this section, we first normalize $\tau$ in the system (\ref{system}) for the purpose of obtaining the smooth dependency of (\ref{system}) related to $\tau$ (i.e., in order to consider the parameter $\tau$ as a bifurcation parameter). By applying the following time-scaling and age-scaling
\begin{equation*}
  \hat{a}=\frac{a}{\tau} \quad \mbox{and} \quad \hat{t}=\frac{t}{\tau},
\end{equation*}
and the following distribution
\begin{equation*}
  \hat{T}(\hat{t})=T(\tau\hat{t}) ,\hat{V}(\hat{t})=V(\tau\hat{t}) \quad \mbox{and} \quad \hat{i}(\hat{t},\hat{a})=\tau i(\tau\hat{t},\tau\hat{a}),
\end{equation*}
after the change of variables and dropping the hat notation, the new system is given by
\begin{equation}\label{newsystem}
\left\{
\begin{array}{l}
\frac{dT(t)}{dt}=\tau\bigg[\Lambda-\mu T(t)+rT(t)\left(1-\frac{T(t)+ \int_{0}^{+\infty}{i(t,a)da}}{K}\right)-\beta_{1} T(t)V(t)\\
\quad\quad\quad-\beta_{2} T(t)\int_{0}^{+\infty}{i(t,a)da}\bigg],\\
\frac{dV(t)}{dt}=\tau\left(N\sigma\int_{0}^{+\infty}{i(t,a)da}-cV(t)\right),\\
\frac{\partial i(t,a)}{\partial t}+\frac{\partial i(t,a)}{\partial a} = -\tau\sigma i(t,a),\\
i(t,0)=\tau \left(\beta_{1} T(t)V(t)+\beta_{2} T(t)\int_{0}^{+\infty}{\beta(a)i(t,a)da}\right),\quad t>0,\\
T(0)=T_{0}\geq0,\quad V(0)=V_{0}\geq0,\quad i(0,\cdot)=i_{0}\in L_{+} ^{1}((0, + \infty ),\mathbb{R}),
\end{array}
\right.
\end{equation}
where the new function $\beta(\cdot)$ is defined by
\begin{equation*}
  \beta(a)=\left\{
             \begin{array}{cc}
               \beta^{*}, & \mbox{  if   } a \geq 1,\\
               0, & \mbox{  otherwise}, \\
             \end{array}
           \right.
\end{equation*}
and
\begin{equation*}
  \int_{\tau}^{+\infty}{\beta^{*}e^{-\sigma a}da}=1\Leftrightarrow \beta^{*}=\sigma e^{\sigma\tau},
\end{equation*}
where $\tau\geq0$, $\beta^{*}>0$.

\noindent

Define $U(t):=\int_{0}^{+\infty}{u(t,a)da}$ in system (\ref{newsystem}), where
$U(t)=\left(
\begin{array}{c}
T(t) \\
V(t) \\
\end{array}
\right)$
and
$u(t,a)=\left(
\begin{array}{c}
\rho(t,a) \\
v(t,a) \\
\end{array}
\right)
$,
the ordinary differential equations in (\ref{newsystem}) can easily be rewritten as an age-structured model
\begin{equation*}
  \left\{
     \begin{array}{ll}
       \frac{\partial u(t,a)}{\partial t}+\frac{\partial u(t,a)}{\partial a} =-\tau C u(t,a),\\
       u(t,0)= \tau G(\rho(t,a),v(t,a)), \\
       u(0,a)=u_{0}\in L^{1}((0,+\infty),\mathbb{R}^{2}), \\
     \end{array}
   \right.
\end{equation*}
where
\begin{equation*}
C=\left(
\begin{array}{cc}
\mu & 0 \\
0 & c \\
\end{array}
\right)
\quad \mbox{and} \quad
    \begin{array}{clc}
        G(\rho(t,a),v(t,a))
      =  \left(
      \begin{array}{c}
      G_{11}\\
      N\sigma\int_{0}^{+\infty}{i(t,a)da} \\
      \end{array}
      \right) &  \\
    \end{array}
\end{equation*}
with
\begin{equation*}
  \begin{array}{ccl}
    G_{11} & = & \Lambda +r\int_{0}^{+\infty}{\rho(t,a)da}\left(1-\frac{\int_{0}^{+\infty}{\rho(t,a)da}+ \int_{0}^{+\infty}{i(t,a)da}}{K}\right)\\
      &&-\beta_{1} \int_{0}^{+\infty}{\rho(t,a)da}\int_{0}^{+\infty}{v(t,a)da}-\beta_{2} \int_{0}^{+\infty}{\rho(t,a)da}\int_{0}^{+\infty}{i(t,a)da}. \\
  \end{array}
\end{equation*}
In what follows, with the notation $w(t,a)=\left(
                           \begin{array}{c}
                             u(t,a) \\
                             i(t,a) \\
                           \end{array}
                         \right)
$, we obtain the equivalent system of model (\ref{newsystem})
\begin{equation}\label{systempartialwta}
\left\{
  \begin{array}{l}
    \frac{\partial w(t,a)}{\partial t}+\frac{\partial w(t,a)}{\partial a} =-\tau Q w(t,a), \\
    w(t,0)=\tau B(w(t,a)), \\
    w(0,\cdot)=w_{0}=\left(
                       \begin{array}{c}
                        \rho_{0} \\
                         v_{0} \\
                         u_{0} \\
                       \end{array}
                     \right)\in L^{1}((0,+\infty),\mathbb{R}^{3}),
     \\
  \end{array}
\right.
\end{equation}
where
\begin{equation*}
      Q=\left(
        \begin{array}{ccc}
          \mu & 0 & 0 \\
          0 & c & 0 \\
          0 & 0 & \sigma \\
        \end{array}
      \right)\\
\end{equation*}
and
\begin{equation*}
    B(w(t,a))=\left(
                                   \begin{array}{c}
                                     G(\rho(t,a),v(t,a)) \\
                                     \beta_{1}\int_{0}^{+\infty}{\rho(t,a)da}\int_{0}^{+\infty}{v(t,a)da}+\beta_{2} \int_{0}^{+\infty}{\rho(t,a)da}\int_{0}^{+\infty}{\beta(a)i(t,a)da} \\
                                   \end{array}
                                 \right).\\
\end{equation*}
\noindent
Subsequently, we consider the following Banach space
\begin{equation*}
  X={\mathbb{R}}^{3} \times L^{1}{((0,+\infty),{\mathbb{R}}^{3})}
\end{equation*}
with $\left \|
\left(
  \begin{array}{c}
    \alpha \\
    \psi\\
  \end{array}
\right)
     \right \|
     =\left \|\alpha\right\|_{{\mathbb{R}}^{3}}+\left\|\psi\right\|_{L^{1}{((0,+\infty),{\mathbb{R}}^{3})}}$.
Define the linear operator $A_{\tau} : D(A_{\tau})\rightarrow X$ by
\begin{equation*}
  A_{\tau}\left(
     \begin{array}{c}
       0_{\mathbb{R}^{3}} \\
       \varphi \\
     \end{array}
   \right)
   =\left(
   \begin{array}{c}
     -\varphi(0) \\
     -\varphi'-\tau Q\varphi \\
   \end{array}
 \right)
\end{equation*}
with $D(A_{\tau})=\{0_{\mathbb{R}^{3}}\}\times W^{1,1}({(0,+\infty),{\mathbb{\mathbb{R}}}^{3}}) \subset X$, and the operator $H: \overline{D(A_{\tau})} \rightarrow X$ by
\begin{equation*}
  H\left(
  \left(
     \begin{array}{c}
       0_{\mathbb{R}^{3}} \\
       \varphi \\
     \end{array}
   \right)
   \right)
   =\left(
      \begin{array}{c}
        B(\varphi) \\
        0_{L^{1}} \\
      \end{array}
    \right).
\end{equation*}
The linear operator $A_{\tau}$ is non-densely defined owing to
\begin{equation*}
  X_{0}:=\overline{D(A_{\tau})}=\{0_{\mathbb{R}^{3}}\} \times L^{1}{((0,+\infty),{\mathbb{R}}^{3})}\neq X.
\end{equation*}
Let
\begin{equation*}
  x(t)=\left(
         \begin{array}{c}
           0_{\mathbb{R}^{3}} \\
           w(t,\cdot) \\
         \end{array}
       \right),
\end{equation*}
system (\ref{systempartialwta}) can be rewritten as the following non-densely defined abstract Cauchy problem
\begin{equation}\label{nonddaCp}
  \left\{
    \begin{array}{l}
      \frac{dx(t)}{dt}  =  A_{\tau}x(t)+\tau H(x(t)),\quad t\geq0, \\
      x(0)  =  \left(
                   \begin{array}{c}
                     0_{\mathbb{R}^{3}} \\
                     w_0 \\
                   \end{array}
                 \right)
                 \in \overline{D(A_{\tau})}. \\
    \end{array}
  \right.
\end{equation}
The global existence and uniqueness of solution of system (\ref{nonddaCp}) follow from the results of \cite{MagalRuan-ADE-2009} and \cite{Magal-EJDE-2001}.
\subsection{Equilibria and linearized equation}

\noindent

In this section, we will discuss the equilibria of the system (\ref{nonddaCp}) and the linearized equation of (\ref{nonddaCp}) around the positive equilibrium.
\subsubsection{Existence of equilibria}

\noindent

Assume that $\overline{x}(a)=\left(
                      \begin{array}{c}
                        0_{\mathbb{R}^{3}} \\
                        \overline{w}(a) \\
                      \end{array}
                    \right)
                    \in X_0
$ is a steady state of system (\ref{nonddaCp}). Then
\begin{equation*}
  \left(
     \begin{array}{c}
       0_{\mathbb{R}^{3}} \\
       \overline{w}(a) \\
     \end{array}
   \right)\in D(A_{\tau})\quad \mbox{and} \quad
   A_{\tau}\left(
      \begin{array}{c}
        0_{\mathbb{R}^{3}} \\
        \overline{w}(a) \\
      \end{array}
    \right)+\tau H\left(\left(
               \begin{array}{c}
                 0_{\mathbb{R}^{3}} \\
                 \overline{w}(a) \\
               \end{array}
             \right)\right)=0,
\end{equation*}
which is equivalent to
\begin{equation*}
\left\{
  \begin{array}{l}
    -\overline{w}(0)+\tau B(\overline{w}(a))=0, \\
    -\overline{w}^{'}(a)-\tau Q\overline{w}(a)=0. \\
  \end{array}
\right.
\end{equation*}
Solving the above equations, we obtain
\begin{equation}\label{overlinewa}
  \left.
    \begin{array}{ccccc}
      \overline{w}(a)  = \left(
                              \begin{array}{c}
                                \overline{\rho}(a) \\
                                 \overline{v}(a) \\
                                 \overline{i}(a) \\
                              \end{array}
                            \right)
        = \left(
               \begin{array}{c}
                 \tau \left[\Lambda +r\overline{T}\left(1-\frac{\overline{T}+\int_{0}^{+\infty}{\overline{i}(a)da}}{K}\right)-\beta_{1} \overline{T}\overline{V} -\beta_{2}\overline{T}\int_{0}^{\infty}{\overline{i}(a)da}\right]e^{-\tau\mu a}\\
                  \tau\left( N\sigma\int_{0}^{+\infty}{\overline{i}(a)da}\right) e^{-\tau c a} \\
                 \tau\left( \beta_{1}\overline{T}\overline{V}+\beta_{2}\overline{T}\int_{0}^{+\infty}{\beta(a)\overline{i}(a)da}\right)e^{-\tau \sigma a}\\
               \end{array}
             \right)
        \\
    \end{array}
  \right.
\end{equation}
with $\overline{T}=\int_{0}^{+\infty}{\overline{\rho}(a)da}$ and $\overline{V}=\int_{0}^{+\infty}{\overline{v}(a)da}$.

\noindent

It follows from the third equation of (\ref{overlinewa}) that
\begin{equation}\label{overlineia}
\overline{i}(a)= \tau\left( \beta_{1} \overline{T}\overline{V} +\beta_{2}\overline{T}\int_{0}^{\infty}{\beta(a)\overline{i}(a)da}\right)e^{-\tau \sigma a}.
\end{equation}
Integrating the equation (\ref{overlineia}) , we have
\begin{equation}\label{betaaia}
  \int_{0}^{+\infty}{\beta(a)\overline{i}(a)da}=\frac{\beta_{1}\overline{T}\overline{V}}{1-\beta_{2}\overline{T}} \quad \mathrm{and} \quad
  \int_{0}^{+\infty}{\overline{i}(a)da}=\frac{1}{\sigma}\int_{0}^{+\infty}{\beta(a)\overline{i}(a)da}.
\end{equation}
It follows from the second equation of (\ref{overlinewa}) that
\begin{equation}\label{overlineV}
  \overline{V}=\int_{0}^{+\infty}{\overline{v}(a)da}=\tau\left( N\sigma\int_{0}^{+\infty}{\overline{i}(a)da}\right)\int_{0}^{+\infty}{e^{-\tau c a }da}=\frac{N}{c}\int_{0}^{+\infty}{\beta(a)\overline{i}(a)da}.\\
\end{equation}
By substituting (\ref{betaaia}) and (\ref{overlineV}) into the first equation of (\ref{overlinewa}), we get
\begin{equation}\label{overlineT}
   \Lambda-\mu \overline{T}+r\overline{T}\left(1-\frac{\overline{T}+\int_{0}^{+\infty}{\overline{i}(a)da}}{K}\right)-\beta_{1} \overline{T}\overline{V} -\beta_{2}\overline{T}\int_{0}^{\infty}{\overline{i}(a)da}=0.
\end{equation}
Solving the above equations (\ref{overlineV}) and (\ref{overlineT}), we obtain
\begin{equation}\label{overlineTV}
  \left\{
   \begin{array}{l}
     \overline{T}_{\pm}=\frac{K(r-\mu)\pm\sqrt{K^{2}(\mu-r)^{2}+4rK\Lambda}}{2r}, \\
     \overline{V}=0, \\
   \end{array}
 \right.
\quad \mathrm{and} \quad
  \left\{
   \begin{array}{l}
     \overline{T}=\frac{c}{N\beta_{1}+c\beta_{2}}, \\
     \overline{V}=\frac{\sigma N\left( K(N\beta_{1}+c\beta_{2})[\Lambda(N\beta_{1}+c\beta_{2})+c(r-\mu)]-c^{2}r\right)}{c(N\beta_{1}+c\beta_{2})[K(N\sigma\beta_{1}+c\beta_{2})+ c r]}. \\
   \end{array}
 \right.
\end{equation}
Therefore, in accordance with (\ref{overlineia}) and (\ref{overlineTV}), we derive the following lemma.
\begin{lemma}
System (\ref{nonddaCp}) has always the  equilibrium
  \begin{equation*}
  \overline{x}_{01}(a)=\left(
                        \begin{array}{c}
                          0_{\mathbb{R}^{3}} \\
                          \left(
                            \begin{array}{c}
                              \tau\mu\overline{T}_{+} e^{-\tau\mu a}\\
                              0_{L^{1}}\\
                              0_{L^{1}} \\
                            \end{array}
                          \right)\\
                        \end{array}
                      \right)
  \quad \mathrm{and} \quad
\overline{x}_{02}(a)=\left(
                        \begin{array}{c}
                          0_{\mathbb{R}^{3}} \\
                          \left(
                            \begin{array}{c}
                              \tau\mu\overline{T}_{-} e^{-\tau\mu a}\\
                              0_{L^{1}}\\
                              0_{L^{1}} \\
                            \end{array}
                          \right)\\
                        \end{array}
                        \right).
\end{equation*}
Furthermore, there exists a unique positive equilibrium of system (\ref{nonddaCp})
\begin{equation*}
\overline{x}_{\tau}=\left(
                  \begin{array}{c}
                    0_{\mathbb{R}^{3}} \\
                    \overline{w}_{\tau} \\
                  \end{array}
                \right)=\left(
  \begin{array}{c}
    0_{\mathbb{R}^{3}} \\
    \left(
      \begin{array}{c}
        \tau\frac{c\mu}{N\beta_{1}+c\beta_{2}}e^{-\tau\mu a} \\
        \tau\frac{\sigma N\left( K(N\beta_{1}+c\beta_{2})[\Lambda(N\beta_{1}+c\beta_{2})+c(r-\mu)]-c^{2}r\right)}{(N\beta_{1}+c\beta_{2})[K(N\sigma\beta_{1}+c\beta_{2})+ c r]}e^{-\tau c a} \\
        \tau\frac{\sigma\left( K(N\beta_{1}+c\beta_{2})[\Lambda(N\beta_{1}+c\beta_{2})+c(r-\mu)]-c^{2}r\right)}{(N\beta_{1}+c\beta_{2})[K(N\sigma\beta_{1}+c\beta_{2})+ c r]}e^{-\tau\sigma a} \\
      \end{array}
    \right)
  \end{array}
\right)
\end{equation*}
if and only if
\begin{equation*}
  K(N\beta_{1}+c\beta_{2})[\Lambda(N\beta_{1}+c\beta_{2})+c(r-\mu)]-c^{2}r>0.
\end{equation*}
Correspondingly, there exists a unique positive equilibrium of system (\ref{system})
\begin{equation*}
  \left(
     \begin{array}{c}
       \overline{T} \\
       \overline{V} \\
       \overline{i}_{\tau}(a) \\
     \end{array}
   \right)=
   \left(
     \begin{array}{c}
       \frac{c\mu}{N\beta_{1}+c\beta_{2}} \\
       \frac{\sigma N\left( K(N\beta_{1}+c\beta_{2})[\Lambda(N\beta_{1}+c\beta_{2})+c(r-\mu)]-c^{2}r\right)}{c(N\beta_{1}+c\beta_{2})[K(N\sigma\beta_{1}+c\beta_{2})+ c r]} \\
       \tau\frac{\sigma\left( K(N\beta_{1}+c\beta_{2})[\Lambda(N\beta_{1}+c\beta_{2})+c(r-\mu)]-c^{2}r\right)}{(N\beta_{1}+c\beta_{2})[K(N\sigma\beta_{1}+c\beta_{2})+ c r]}e^{-\tau\sigma a}\\
     \end{array}
   \right)
\end{equation*}
if and only if
\begin{equation*}
  K(N\beta_{1}+c\beta_{2})[\Lambda(N\beta_{1}+c\beta_{2})+c(r-\mu)]-c^{2}r>0.
\end{equation*}
\end{lemma}
In the following, we always assume that $K(N\beta_{1}+c\beta_{2})[\Lambda(N\beta_{1}+c\beta_{2})+c(r-\mu)]-c^{2}r>0$.

\subsubsection{Linearized equation}

\noindent

In order to obtain the linearized equation of (\ref{nonddaCp}) around the positive equilibrium $\overline{x}_{\tau}$, we first make the following change of variable
\begin{equation*}
  y(t):= x(t)-\overline{x}_{\tau},
\end{equation*}
and then, (\ref{nonddaCp}) becomes
\begin{equation}\label{system4}
  \left\{
    \begin{array}{cll}
      \frac{dy(t)}{dt} & = & A_{\tau}y(t)+\tau H(y(t)+\overline{x}_{\tau})-\tau H(\overline{x}_{\tau}), t\geq0, \\
      y(0) & = & \left(
                   \begin{array}{c}
                     0_{\mathbb{R}^{3}} \\
                     w_0-\overline{w}_{\tau} \\
                   \end{array}
                 \right)
                 =: y_0\in \overline{D(A_{\tau})}.
       \\
    \end{array}
  \right.
\end{equation}
Therefore the linearized equation (\ref{system4}) around the equilibrium $0$ is given by
\begin{equation}\label{systemlinear}
  \begin{array}{cc}
    \frac{dy(t)}{dt}=A_{\tau}y(t)+\tau DH(\overline{x}_{\tau})y(t) &\quad\mbox{for}\quad t\geq 0,\quad y(t)\in X_{0}, \\
  \end{array}
\end{equation}
where
\begin{equation*}
    \begin{array}{cc}
      \tau DH(\overline{x}_{\tau})\left(
                        \begin{array}{c}
                          0_{\mathbb{R}^{3}} \\
                          \varphi \\
                        \end{array}
                      \right)=\left(
                                \begin{array}{c}
                                 \tau DB(\overline{w}_{\tau})(\varphi) \\
                                  0_{L^{1}} \\
                                \end{array}
                              \right)
       &\quad\mbox{for all}\quad\left(
                         \begin{array}{c}
                           0_{\mathbb{R}^{3}} \\
                           \varphi \\
                         \end{array}
                       \right)\in D(A_{\tau})\\
    \end{array}
\end{equation*}
with
\begin{equation*}
     \begin{array}{ccl}
       DB(\overline{w}_{\tau})(\varphi) & = & \left(
                              \begin{array}{ccc}
                                r-\frac{2r\overline{T}}{K}-(\frac{r}{K}+\beta_{2})\int_{0}^{+\infty}{\overline{i}(a)da}-\beta_{1}\overline{V} & -\beta_{1}\overline{T}&-(\frac{r}{K}+\beta_{2})\overline{T}\\
                                0&0&N\sigma\\
                                \beta_{1}\overline{V}+\beta_{2}\int_{0}^{+\infty}{\beta(a)\overline{i}(a)da}&\beta_{1}\overline{T}&0\\
                              \end{array}
                            \right)\\
                            &&\times\int_{0}^{+\infty}{\varphi(a)da}
         +\left(
                              \begin{array}{ccc}
                                0&0&0 \\
                                0&0&0\\
                                0&0&\beta_{2}\overline{T}\\
                              \end{array}
                            \right)
        \int_{0}^{+\infty}{\beta(a)\varphi(a)da}. \\
     \end{array}
\end{equation*}
Then we can rewrite system (\ref{system4}) as
\begin{equation}\label{fracdytdt}
    \begin{array}{cc}
      \frac{dy(t)}{dt}=B_{\tau}y(t)+\mathcal{H}(y(t)) &\quad \mbox{for}\quad t\geq0, \\
    \end{array}
\end{equation}
where
\begin{equation*}
  B_{\tau}:=A_{\tau}+\tau DH(\overline{x}_{\tau})
\end{equation*}
is a linear operator and
\begin{equation*}
  \mathcal{H}(y(t))=\tau H(y(t)+\overline{x}_{\tau})-\tau H(\overline{x}_{\tau})-\tau DH(\overline{x}_{\tau})y(t)
\end{equation*}
satisfying $\mathcal{H}(0)=0$ and $D\mathcal{H}(0)=0$.

\subsection{Characteristic equation}

\noindent

In this section, we will obtain the characteristic equation of (\ref{nonddaCp}) around the positive equilibrium $\overline{x}_{\tau}$.
Denote
\begin{equation*}
  \nu:=\min\{\mu, c, \sigma\}>0 \quad\mbox{and}\quad \Omega := \{\lambda \in \mathbb{C} : Re(\lambda)>-\nu\tau\}.
\end{equation*}
Applying the results of \cite{LiuMagalRuan-ZAMP-2011}, we obtain the following result.
\begin{lemma}
For $\lambda\in \Omega$, $\lambda\in \rho(A_{\tau})$ and
\begin{equation*}
  (\lambda I-A_{\tau})^{-1}\left(
                               \begin{array}{c}
                                 \delta \\
                                 \psi\\
                               \end{array}
                             \right)
                             =\left(
                                \begin{array}{c}
                                  0_{\mathbb{R}^{3}} \\
                                  \varphi \\
                                \end{array}
                              \right)
                              \Leftrightarrow
                              \varphi(a)=e^{-\int_{0}^{a}{(\lambda I+\tau Q)dl}}\delta+\int_{0}^{a}{e^{-\int_{s}^{a}{(\lambda I+\tau Q)dl}}\psi(s)}ds
\end{equation*}
with $\left(
        \begin{array}{c}
          \delta \\
          \psi \\
        \end{array}
      \right)
      \in X
$ and $\left(
         \begin{array}{c}
           0_{\mathbb{R}^{3}} \\
           \varphi \\
         \end{array}
       \right)
       \in D(A_{\tau})
$. Moreover, $A_{\tau}$ is a Hille-Yosida operator and
\begin{equation}\label{Hille-Yosida}
  \left\|(\lambda I-A_{\tau})^{-n}\right\|\leq\frac{1}{(Re(\lambda)+\nu\tau)^{n}},\quad\forall \lambda\in\Omega,\quad\forall n\geq 1.
\end{equation}
\end{lemma}
\noindent
Let $A_0$ be the part of $A_{\tau}$ in $\overline{D(A_{\tau})}$, that is, $A_0 := D(A_0)\subset X \rightarrow X$. For $\left(
                                                                                                    \begin{array}{c}
                                                                                                      0_{\mathbb{R}^{3}} \\
                                                                                                      \varphi \\
                                                                                                    \end{array}
                                                                                                  \right)
                                                                                                  \in D(A_0)
$, we have
\begin{equation*}
  A_0\left(
       \begin{array}{c}
         0_{\mathbb{R}^{3}} \\
         \varphi \\
       \end{array}
     \right)
     =\left(
        \begin{array}{c}
          0_{\mathbb{R}^{3}} \\
          \hat{A_0}(\varphi) \\
        \end{array}
      \right),
\end{equation*}
where $\hat{A_0}(\varphi)=-\varphi '-\tau Q\varphi$ with $D(\hat{A_0})=\{\varphi \in W^{1,1}((0,+\infty),{\mathbb{R}}^{3}): \varphi(0)=0\}$.

\noindent

Note that $\tau DH(\overline{x}_{\tau}):D(A_{\tau}) \subset X \rightarrow X$ is a compact bounded linear operator. Based on (\ref{Hille-Yosida}) we have
\begin{equation*}
  \left\| T_{A_0}(t) \right\| \leq e^{-\nu\tau t} \quad\mbox{for}\quad t \geq 0.
\end{equation*}
Furthermore, we get
\begin{equation*}
  \omega_{0,ess}(A_0)\leq\omega_0(A_{0})\leq -\nu\tau.
\end{equation*}
Using the perturbation results developed in \cite{DucrotLiuMagal-JMAA-2008}, we obtain
\begin{equation*}
  \omega_{0,ess}((A_{\tau}+\tau DH(\overline{x}_{\tau}))_{0})\leq-\nu\tau<0.
\end{equation*}
Hence we conclude the following proposition.
\begin{lemma}
The linear operator $B_{\tau}$ is a Hille-Yosida operator, and its parts $(B_{\tau})_{0}$ in
$\overline{D(B_{\tau})}$ satisfies
\begin{equation*}
  \omega_{0,ess}((B_{\tau})_{0})<0.
\end{equation*}
\end{lemma}
\noindent
Let $\lambda\in \Omega$. Since $(\lambda I-A_{\tau})$ is invertible, and
\begin{equation}\label{invertible}
  \begin{array}{ccl}
    (\lambda I-B_{\tau})^{-1} & = & (\lambda I-(A_{\tau}+\tau DH(\overline{x}_{\tau})))^{-1} \\
     & = & (\lambda I-A_{\tau})^{-1}(I-\tau DH(\overline{x}_{\tau})(\lambda I-A_{\tau})^{-1})^{-1}, \\
  \end{array}
\end{equation}
$\lambda I-B_{\tau}$ is invertible if and only if $I-\tau DH(\overline{x}_{\tau})(\lambda I-A_{\tau})^{-1}$ is invertible.
Set
\begin{equation*}
  (I-\tau DH(\overline{x}_{\tau})(\lambda I-A_{\tau})^{-1})\left(
     \begin{array}{c}
       \delta \\
       \varphi \\
     \end{array}
   \right)
   =\left(
      \begin{array}{c}
        \gamma \\
        \psi \\
      \end{array}
    \right).
\end{equation*}
That is
\begin{equation*}
  \left(
    \begin{array}{l}
      \delta \\
      \varphi \\
    \end{array}
  \right)
  -\tau DH(\overline{x}_{\tau})(\lambda I-A_{\tau})^{-1}\left(
                                                \begin{array}{c}
                                                  \delta \\
                                                  \varphi \\
                                                \end{array}
                                              \right)
                                              =\left(
                                                 \begin{array}{c}
                                                   \gamma \\
                                                   \psi \\
                                                 \end{array}
                                               \right).
\end{equation*}
It follows that
\begin{equation*}
 \left\{
    \begin{array}{l}
      \delta-\tau DB(\overline{w}_{\tau})\left(e^{-\int_{0}^{a}{(\lambda I+\tau Q)dl}}\delta+\int_{0}^{a}{e^{-\int_{s}^{a}{(\lambda I+\tau Q)dl}}\varphi(s)}ds\right)=\gamma, \\
      \varphi=\psi, \\
    \end{array}
  \right.
\end{equation*}
i.e.,
\begin{equation*}
  \left\{
    \begin{array}{l}
 \delta-\tau DB(\overline{w}_{\tau})\left(e^{  -\int_{0}^{a}{(\lambda I+\tau Q)dl}}\delta\right)=\gamma+\tau DB(\overline{w}_{\tau})\left(\int_{0}^{a}{e^{-\int_{s}^{a}{(\lambda I+\tau Q)dl}}\varphi(s)}ds\right), \\
      \varphi=\psi. \\
    \end{array}
  \right.
\end{equation*}
Combining with the formula of $DB(\overline{w}_{\tau})$ we conclude that
\begin{equation*}
  \left\{
    \begin{array}{l}
      \Delta(\lambda)\delta=\gamma+K(\lambda,\psi), \\
      \varphi=\psi, \\
    \end{array}
  \right.
\end{equation*}
where
\begin{equation}\label{Deltalambda}
   \begin{array}{ccl}
     \Delta(\lambda) & = & I-\left(
                              \begin{array}{ccc}
                                r-\frac{2r\overline{T}}{K}-(\frac{r}{K}+\beta_{2})\int_{0}^{+\infty}{\overline{i}(a)da}-\beta_{1}\overline{V} & -\beta_{1}\overline{T}&-(\frac{r}{K}+\beta_{2})\overline{T}\\
                                0&0&N\sigma\\
                                \beta_{1}\overline{V}+\beta_{2}\int_{0}^{+\infty}{\beta(a)\overline{i}(a)da}&\beta_{1}\overline{T}&0\\
                              \end{array}
                            \right)\\
                            &&\times\tau\int_{0}^{+\infty}{e^{-\int_{0}^{a}{(\lambda I+\tau Q)dl}}da} \\
      &  &  -\left(
      \begin{array}{ccc}
                                0&0&0 \\
                                0&0&0\\
                                0&0&\beta_{2}\overline{T}\\
                              \end{array}
                            \right)\tau
                             \int_{0}^{+\infty}{\beta(a)e^{-\int_{0}^{a}{(\lambda I+\tau Q)dl}}}da \\
   \end{array}
\end{equation}
and
\begin{equation}\label{Klambdapsi}
  K(\lambda,\psi)=\tau DB(\overline{w}_{\tau})\left(\int_{0}^{a}{e^{-\int_{s}^{a}{(\lambda I+\tau Q)dl}}\psi(s)}ds\right).
\end{equation}
Whenever $\Delta(\lambda)$ is invertible, we have
\begin{equation}\label{xi}
  \delta=(\Delta(\lambda))^{-1}(\gamma+K(\lambda,\psi)).
\end{equation}
Following the above discussion and the proof of Lemma 3.5 in \cite{WangLiu-JMAA-2012}, we derive the lemma as follows.
\begin{lemma}
The following results hold
\begin{itemize}
  \item [(i)] $\sigma(B_{\tau})\cap\Omega=\sigma_{p}(B_{\tau})\cap\Omega=\{\lambda\in\Omega: \det(\Delta(\lambda))=0\}$;
  \item [(ii)] If $\lambda\in\rho(B_{\tau})\cap\Omega$, we have the following formula for resolvent
  \begin{equation}\label{lambdalambdaI}
    (\lambda I -B_{\tau})^{-1}\left(
                                \begin{array}{c}
                                  \delta \\
                                  \varphi \\
                                \end{array}
                              \right)
                              =\left(
                                 \begin{array}{c}
                                   0_{\mathbb{R}^{2}} \\
                                   \psi \\
                                 \end{array}
                               \right),
  \end{equation}
  where
  \begin{equation*}
    \psi(a)= e^{-\int_{0}^{a}{(\lambda I+\tau Q)dl}}(\Delta(\lambda))^{-1}\left[\gamma+K(\lambda,\varphi)\right]+\int_{0}^{a}
    {e^{-\int_{s}^{a}{(\lambda I+\tau Q)dl}}}\varphi(s)ds
  \end{equation*}
  with $\Delta(\lambda)$ and $K(\lambda,\varphi)$ defined in (\ref{Deltalambda}) and (\ref{Klambdapsi}).
\end{itemize}
\end{lemma}
\noindent
Under Assumption \ref{assumption1}, we have
\begin{equation}\label{intea}
  \int_{0}^{+\infty}{e^{-\int_{0}^{a}({\lambda I+\tau Q})dl}}da=
  \left(
    \begin{array}{ccc}
      \frac{1}{\lambda+\mu\tau} & 0 &0\\
      0 & \frac{1}{\lambda+c\tau}&0 \\
      0&0& \frac{1}{\lambda+\sigma\tau}
    \end{array}
  \right)
\end{equation}
and
\begin{equation}\label{intbetaea}
  \int_{0}^{+\infty}{\beta(a)e^{-\int_{0}^{a}({\lambda I+\tau Q})dl}}da=
  \left(
    \begin{array}{ccc}
      \frac{\beta^{*}e^{-(\lambda+\mu\tau)}}{\lambda+\mu\tau} & 0 &0\\
      0 & \frac{\beta^{*}e^{-(\lambda+c\tau)}}{\lambda+c\tau}&0 \\
      0&0&\frac{\beta^{*}e^{-(\lambda+\sigma\tau)}}{\lambda+\sigma\tau}
    \end{array}
  \right).
\end{equation}
It follows from (\ref{Deltalambda}), (\ref{intea}) and (\ref{intbetaea}) that the characteristic equation at the positive equilibrium $\overline{x}_{\tau}$ is
\begin{equation}\label{characteristicequation}
     \begin{array}{ccl}
       \det(\Delta(\lambda)) & = & \left|
  \begin{array}{ccc}
                  \tau\frac{\beta_{1}\overline{V}+\frac{2r\overline{T}}{K}+\left(\beta_{2}+\frac{r}{K}\right)\frac{\xi}{\sigma}-r}{\mu\tau+\lambda}+1& \frac{\tau\beta_{1}\overline{T}}{c\tau+\lambda}&\frac{\tau\overline{T}\left(\beta_{2}+\frac{r}{K}\right)}{\sigma\tau+\lambda} \\
                  0&1&-\frac{\tau N\sigma}{\sigma\tau+\lambda}\\
                  -\tau\frac{\beta_{1}\overline{V}+\beta_{2}\xi}{\mu\tau+\lambda}&-\frac{\tau\beta_{1}\overline{T}}{c\tau+\lambda}&-\frac{\tau\beta_{2}\sigma\overline{T} e^{-\lambda}}{\sigma\tau+\lambda}+1\\
                 \end{array}
  \right| \\
  & = & \frac{ \lambda^{3}+\tau p_{2}\lambda^{2}+\tau^{2} p_{1}\lambda+\tau^{3}p_{0}+(\tau q_{2}\lambda^{2} +\tau^{2} q_{1}\lambda +\tau^{3}q_{0})e^{-\lambda}}{(\lambda+\tau\mu)(\lambda+\tau c)(\lambda+\tau\sigma)} \\
        & \triangleq & \frac{\tilde{f}(\lambda)}{\tilde{g}(\lambda)}=0, \\
     \end{array}
\end{equation}
where
\begin{equation*}
     \begin{array}{ccl}
       \overline{T}& = &\int_{0}^{+\infty}{\overline{\rho}(a)da},\\
       \overline{V}& = &\int_{0}^{+\infty}{\overline{v}(a)da},\\
       \xi & = & \int_{0}^{+\infty}{\beta(a)\overline{i}(a)da}=\sigma\int_{0}^{+\infty}{\overline{i}(a)da}, \\
       p_{2} & = & \beta_{1}\overline{V}+\mu+c+\sigma-r+\frac{\beta_{2}\xi}{\sigma}+\frac{2r\overline{T}}{K}+\frac{r\xi}{K\sigma}, \\

       p_{1} & = & \beta_{1}\overline{T}(\beta_{2}\overline{V}-\sigma N)+\beta_{2}\xi(\beta_{2}\overline{T}+1)
       +(c+\sigma)(\beta_{1}\overline{V}+\mu-r)+c\sigma \\
       &&+\frac{r}{K}[\overline{T}(\beta_{1}\overline{V}+\beta_{2}\xi)+2\overline{T}(c+\sigma)+\xi]+\frac{c\xi}{K\sigma}(K\beta_{2}+r),\\

       p_{0} & = & N\beta_{1}\beta_{2}\xi\overline{T}(\sigma-1)+\sigma(r-\mu)(N\beta_{1}\overline{T}-c)
       +c\beta_{2}\overline{T}(\beta_{1}\overline{V}+\beta_{2}\xi)\\
       &&+c(\beta_{1}\sigma\overline{V}+\beta_{2}\xi)
       -\frac{r}{K}[(N\beta_{1}\overline{T}-c)(2\sigma\overline{T}+\xi)- c\overline{T}(\beta_{1}\overline{V}+\beta_{2}\xi)],\\

       q_{2} & = & -\sigma\beta_{2}\overline{T}, \\

       q_{1} & = & -\beta_{2}\overline{T}\big[\sigma(\beta_{1}\overline{V}+\mu+c-r)+\beta_{2}\xi+\frac{r}{K}(\xi+2\sigma\overline{T})\big], \\

       q_{0} & = & -c\beta_{2}\overline{T}[\sigma(\beta_{1}\overline{V}+\mu-r)+\beta_{2}\xi+\frac{r}{K}(2\sigma\overline{T}+\xi)], \\
       \tilde{f}(\lambda) & = & \lambda^{3}+\tau p_{2}\lambda^{2}+\tau^{2} p_{1}\lambda+\tau^{3}p_{0}+(\tau q_{2}\lambda^{2} +\tau^{2} q_{1}\lambda +\tau^{3}q_{0})e^{-\lambda}, \\

       \tilde{g}(\lambda) & = & (\lambda+\tau\mu)(\lambda+\tau c)(\lambda+\tau\sigma). \\
     \end{array}
\end{equation*}
Let
\begin{equation*}
  \lambda=\tau\zeta.
\end{equation*}
Then we get
\begin{equation}\label{characteristicequationg}
  \tilde{f}(\lambda)=\tilde{f}(\tau\zeta):=\tau^{3}g(\zeta)=\tau^{3}\big[\zeta^{3}+ p_{2}\zeta^{2}+p_{1}\zeta+ p_{0}+(q_{2}\zeta^{2} +q_{1}\zeta +q_{0})e^{-\tau\zeta}\big].
\end{equation}
It is straightforward to demonstrate that
\begin{equation*}
  \{\lambda\in\Omega:\det(\Delta(\lambda))=0\}=\{\lambda=\tau\zeta\in\Omega:g(\zeta)=0\}.
\end{equation*}

\section{Existence of Hopf bifurcation}

\noindent

In this section, the parameter $\tau$ will be viewed as a Hopf bifurcation parameter and the existence of Hopf bifurcation for the Cauchy problem (\ref{nonddaCp}) will be further investigated by applying the Hopf bifurcation theory \cite{LiuMagalRuan-ZAMP-2011}.
On the basis of (\ref{characteristicequationg}), we have
\begin{equation}\label{characteristicequation3}
 g(\zeta)=\zeta^{3}+ p_{2}\zeta^{2}+p_{1}\zeta+ p_{0}+(q_{2}\zeta^{2}+q_{1}\zeta +q_{0})e^{-\tau\zeta},
\end{equation}
where
\begin{equation}\label{p1p2p3p4p5p6}
    \begin{array}{ccl}
      p_{2} & = & \frac{K(N\beta_{1}+c\beta_{2})[\Lambda(N\beta_{1}+c\beta_{2})+c(c+\sigma)]+c^{2}r}{cK(N\beta_{1}+c\beta_{2})}, \\

      p_{1} & = & \frac{\mathcal{P}}{cK(N\beta_{1}+c\beta_{2})[K(N\beta_{1}\sigma+c\beta_{2})+ c r]}, \\

       p_{0} & = & \frac{\sigma\left(K(N\beta_{1}+c\beta_{2})^{2}[\Lambda(N\beta_{1}+2c\beta_{2})+c(r-\mu)]-N\beta_{1}c^{2}r\right)}{K(N\beta_{1}+c\beta_{2})^{2}}, \\

      q_{2} & = & -\frac{c\sigma\beta_{2}}{N\beta_{1}+c\beta_{2}}, \\

      q_{1} & = & -\frac{\sigma \beta_{2}\left(K(N\beta_{1}+c\beta_{2})[\Lambda(N\beta_{1}+c\beta_{2})+c^{2}]+c^{2}r\right)}
      {K(N\beta_{1}+c\beta_{2})^{2}}, \\

      q_{0} & = & -\frac{c\sigma \beta_{2}[K\Lambda(N\beta_{1}+c\beta_{2})^{2}+c^{2}r]}{K(N\beta_{1}+c\beta_{2})^{2}}, \\
    \end{array}
\end{equation}
and
\begin{equation*}
\begin{array}{ccl}
  \mathcal{P} &=& K^{2}\Lambda(N\beta_{1}+c\beta_{2})^{2}[N\beta_{1}\sigma(c+\sigma)+c\beta_{2}(c+2\sigma)]+K^{2}N\beta_{1}\beta_{2}c^{2}\sigma(c\sigma+r-\mu)\\
      &&+K^{2}\beta_{2}^{2}c^{3}\sigma(c+r-\mu)+K\Lambda c r(N\beta_{1}+c\beta_{2})^{2}(c+2\sigma)+KN\beta_{1}c^{2}\sigma r(c+\sigma+r-\mu)\\
      &&+K\beta_{2}c^{3}r[(c+r-\mu)\sigma+c)]+c^{4}r^{2}.\\
\end{array}
\end{equation*}
In addition,
\begin{equation}\label{p0q0}
  \begin{array}{ccl}
    p_{0}+q_{0} & = & \frac{\sigma\left( K(N\beta_{1}+c\beta_{2})[\Lambda(N\beta_{1}+c\beta_{2})+c(r-\mu)]
    -c^{2}r\right)}{K(N\beta_{1}+c\beta_{2})}, \\
  \end{array}
\end{equation}
If $K(N\beta_{1}+c\beta_{2})[\Lambda(N\beta_{1}+c\beta_{2})+c(r-\mu)]-c^{2}r>0$, then $p_{0}+q_{0}>0$, $p_{0}-q_{0}>0$ and
$\zeta =0$ is not a eigenvalue of (\ref%
{characteristicequation3}).

In what follows, we first let $\zeta=i\omega (\omega>0)$ be a purely imaginary root of $g(\zeta)=0$, that is,
\begin{equation*}
  -i\omega^{3}-p_{2}\omega^{2}+ip_{1}\omega+p_{0}+(-q_{2}\omega^{2}+iq_{1}\omega+q_{0})e^{-i\omega\tau}=0.
\end{equation*}
Separating the real part and the imaginary part in the above equation, we have
\begin{equation}\label{realimaginary}
  \left\{
     \begin{array}{l}
       p_{2}\omega^{2}-p_{0}=(q_{0}-q_{2}\omega^{2})\cos(\omega\tau)+q_{1}\omega\sin(\omega\tau), \\
       p_{1}\omega-\omega^{3}=(q_{0}-q_{2}\omega^{2})\sin(\omega\tau)-q_{1}\omega\cos(\omega\tau). \\
     \end{array}
   \right.
\end{equation}
 Consequently, we can further obtain
\begin{equation*}
  (p_{2}\omega^{2}-p_{0})^{2}+(p_{1}\omega-\omega^{3})^{2}=(q_{0}-q_{2}\omega^{2})^2+(q_{1}\omega)^{2},
\end{equation*}
i.e.,
\begin{equation}\label{omega34}
  \omega^{6}+(p_{2}^{2}-q_{2}^{2}-2p_{1})\omega^{4}+(p_{1}^{2}-q_{1}^{2}-2p_{2}p_{0}+2q_{2}q_{0})\omega^{2}+p_{0}^{2}-q_{0}^{2}=0.
\end{equation}
Set $\omega^{2}=\theta$. Now (\ref{omega34}) becomes
\begin{equation}\label{theta2}
 \theta^{3}+C_{2}\theta^{2}+C_{1}\theta+C_{0}=0,
\end{equation}
where
\begin{equation}\label{C2C1C0}
     \begin{array}{ccl}
       C_{2} & = & p_{2}^{2}-q_{2}^{2}-2p_{1}, \\
       C_{1} & = & p_{1}^{2}-q_{1}^{2}-2p_{2}p_{0}+2q_{2}q_{0}, \\
       C_{0} & = & p_{0}^{2}-q_{0}^{2}. \\
     \end{array}
\end{equation}
Let $\theta_{1}$, $\theta_{2}$ and $\theta_{3}$ denote the three roots of (\ref{theta2}). According to the theorem of
Vieta, we have
\begin{equation*}
  \theta_{1}+\theta_{2}+\theta_{3}=-C_{2}\quad\mbox{and}\quad \theta_{1}\theta_{2}\theta_{3}=-C_{0}.
\end{equation*}
Combing with (\ref{p0q0}) and (\ref{C2C1C0}), we can get that
\begin{equation}\label{theta1theta2theta3}
  \theta_{1}\theta_{2}\theta_{3}=-C_{0}=-(p_{0}+q_{0})(p_{0}-q_{0})<0.
\end{equation}
Denote
\begin{equation}\label{D}
  D:=p^{3}+q^{2},
\end{equation}
where
\begin{equation*}
       q  =  \frac{C_{2}^{3}}{27}-\frac{C_{2}C_{1}}{6}+\frac{C_{0}}{2} \quad\mbox{and}\quad
       p  =  \frac{C_{1}}{3}-\frac{C_{2}^{2}}{9}.
\end{equation*}
The quantity (\ref{D}) is named the discriminant of (\ref{theta2}). Table \ref{Cubicequations} describes the
behavior of the solutions of (\ref{theta2}) under the condition that the coefficients are real.
\begin{table}
\caption{Cubic equations with real coefficients.}\label{Cubicequations}
\centering
\doublerulesep=0.4pt
\begin{tabular}{p{4cm}p{7cm}}
\hline
\multicolumn{1}{l}{Discriminant} & \multicolumn{1}{l}{Description}\\ \hline
$D>0$ & {\footnotesize one real and two conjugate complex zeros.}\\
$D<0$ & {\footnotesize three distinct real zeros.}\\
$D=0, q\neq0$ & {\footnotesize two real zeros, one of which is double.}\\
$D=0, q=0$ & {\footnotesize one triple real zero.}\\
\hline
\end{tabular}
\end{table}
The solutions of (\ref{theta2}) can be given by
\begin{equation}\label{solutiongamma}
     \begin{array}{ccl}
       \theta_{1} & = & \sqrt[3]{-q+\sqrt{D}}+\sqrt[3]{-q-\sqrt{D}}, \\
       \theta_{2} & = & \frac{-1+i\sqrt{3}}{2}\sqrt[3]{-q+\sqrt{D}}+\frac{-1-i\sqrt{3}}{2}\sqrt[3]{-q-\sqrt{D}}, \\
       \theta_{3} & = & \frac{-1-i\sqrt{3}}{2}\sqrt[3]{-q+\sqrt{D}}+\frac{-1+i\sqrt{3}}{2}\sqrt[3]{-q-\sqrt{D}}. \\
     \end{array}
\end{equation}
And then, it follows from (\ref{theta1theta2theta3}), Table \ref{Cubicequations} and (\ref{solutiongamma}) that when $D=0$ and $q>0$, (\ref{theta2}) has only one double positive real root $\theta_{0}$. Therefore (\ref{omega34}) has only one positive real root $\omega_{0}=\sqrt{\theta_{0}}$. On the basis of (\ref{realimaginary}), we can further conclude that $g(\zeta)=0$ with $\tau=\tau_{k}$, $k=1,2,\cdots$ has a pair of purely imaginary roots $\pm i\omega_{0}$, where
\begin{equation*}
  \omega_{0}=\sqrt{-\sqrt[3]{-q}}
\end{equation*}
and
\begin{equation}\label{tauk}
  \tau_{k}=\left\{
             \begin{array}{l}
               \frac{1}{\omega_{0}}\left\{\arccos\left(-\frac{(p_{2}q_{2}-q_{1})\omega_{0}^{4}-(p_{2}q_{0}+p_{0}q_{2}-p_{1}q_{1})\omega_{0}^{2}+p_{0}q_{0}}
               {q_{2}^{2}\omega_{0}^{4}+(q_{1}^{2}-2q_{2}q_{0})\omega_{0}^{2}+q_{0}^{2}}\right)+2k\pi\right\},
               \mbox{  if  } \Theta\geq 0,\\
                \frac{1}{\omega_{0}}\left\{2\pi-\arccos\left(-\frac{(p_{2}q_{2}-q_{1})\omega_{0}^{4}-(p_{2}q_{0}+p_{0}q_{2}-p_{1}q_{1})\omega_{0}^{2}+p_{0}q_{0}}
               {q_{2}^{2}\omega_{0}^{4}+(q_{1}^{2}-2q_{2}q_{0})\omega_{0}^{2}+q_{0}^{2}}\right)+2k\pi\right\},
                \mbox{  if  } \Theta< 0,\\
             \end{array}
           \right.
\end{equation}
for $k=1,2,\cdots$ with
\begin{equation*}
  \Theta:=-\frac{\omega_{0}[q_{2}\omega_{0}^{4}+(p_{2}q_{1}-p_{1}q_{2}-q_{0})\omega_{0}^{2}+p_{1}q_{0}-p_{0}q_{1}]}
               {q_{2}^{2}\omega_{0}^{4}+(q_{1}^{2}-2q_{2}q_{0})\omega_{0}^{2}+q_{0}^{2}}.
\end{equation*}
\begin{assumption}\label{assumption2}
  Assume that $K(N\beta_{1}+c\beta_{2})[\Lambda(N\beta_{1}+c\beta_{2})+c(r-\mu)]-c^{2}r>0$, $D=0$, $q>0$, $C_{2}>0$ and $C_{1}>0$ where $D$, $q$, $C_{2}$ and $C_{1}$ are given by (\ref{D}) and (\ref{C2C1C0}).
\end{assumption}
\begin{lemma}
Let Assumption \ref{assumption1} and \ref{assumption2} hold, then
\begin{equation*}
  \frac{dg(\zeta)}{d\zeta}\Big|_{\zeta=i\omega_{0}}\neq0.
\end{equation*}
Therefore $\zeta=i\omega_{0}$ is a simple root of (\ref{characteristicequation3}).
\end{lemma}
\begin{proof}
On the basis of (\ref{characteristicequation3}), we obtain
\begin{equation*}
 \frac{dg(\zeta)}{d\zeta}\Big|_{\zeta=i\omega_{0}}=\left\{3\zeta^{2}+2p_{2}\zeta+p_{1}+[2q_{2}\zeta+q_{1}
 -\tau(q_{2}\zeta^{2}+q_{1}\zeta+q_{0})]e^{-\tau\zeta}\right\}\Big|_{\zeta=i\omega_{0}}
\end{equation*}
and
\begin{equation*}
\left\{3\zeta^{2}+2p_{2}\zeta+p_{1}+[2q_{2}\zeta+q_{1}
 -\tau(q_{2}\zeta^{2}+q_{1}\zeta+q_{0})]e^{-\tau\zeta}\right\}\frac{d\zeta(\tau)}{d\tau}
 =\zeta(q_{2}\zeta^{2} +q_{1}\zeta+q_{0})e^{-\tau\zeta}.
\end{equation*}
Suppose that $\frac{dg(\zeta)}{d\zeta}\Big|_{\zeta=i\omega_{0}}=0$, then
\begin{equation*}
  i\omega_{0}(-q_{2}\omega_{0}^{2}+iq_{1}\omega_{0}+q_{0})e^{-i\omega_{0}\tau}=0.
\end{equation*}
Separating real and imaginary parts in the above equation, we have
\begin{equation}
  \left\{
     \begin{array}{l}
       (q_{0}\omega_{0}-q_{2}\omega_{0}^{3})\sin(\omega_{0}\tau)-q_{1}\omega_{0}^{2}\cos(\omega_{0}\tau)=0, \\
       (q_{0}\omega_{0}-q_{2}\omega_{0}^{3})\cos(\omega_{0}\tau)+q_{1}\omega_{0}^{2}\sin(\omega_{0}\tau)=0. \\
     \end{array}
   \right.
\end{equation}
That is,
\begin{equation*}
  (q_{0}\omega_{0}-q_{2}\omega_{0}^{3})^{2}+(q_{1}\omega_{0}^{2})^{2}=0,
\end{equation*}
which implies
\begin{equation*}
  q_{0}\omega_{0}-q_{2}\omega_{0}^{3}=0 \quad \mathrm{and}\quad q_{1}\omega_{0}^{2}=0.
\end{equation*}
Since $\omega_{0}>0$, it follows that
\begin{equation*}
  q_{0}-q_{2}\omega_{0}^{2}=0 \quad \mathrm{and}\quad q_{1}=0.
\end{equation*}
However, $q_{1}=-\frac{\sigma \beta_{2}\left(K(N\beta_{1}+c\beta_{2})[\Lambda(N\beta_{1}+c\beta_{2})+c^{2}]+c^{2}r\right)}
      {K(N\beta_{1}+c\beta_{2})^{2}}<0$ which leads to a contradiction. Hence
\begin{equation*}
  \frac{dg(\zeta)}{d\zeta}\Big|_{\zeta=i\omega_{0}}\neq0.
\end{equation*}
This completes the proof.
\end{proof}
\begin{lemma}
  Let Assumption \ref{assumption1} and \ref{assumption2} hold. Denote the root $\zeta(\tau)=\alpha(\tau)+i\omega(\tau)$ of $g(\zeta)=0$ satisfying $\alpha(\tau_{k})=0$ and $\omega(\tau_{k})=\omega_{0}$, where $\tau_{k}$ is defined in (\ref{tauk}). Then
\begin{equation*}
  \alpha^{'}(\tau_{k})=\frac{\mbox{d}Re(\zeta)}{\mbox{d}\tau}\Big|_{\tau=\tau_{k}}>0.
\end{equation*}
\end{lemma}
\begin{proof}
For simplicity, we discuss $\frac{d\tau}{d\zeta}$ instead of $\frac{d\zeta}{d\tau}$. Based on the expression of $g(\zeta)=0$, we obtain
\begin{equation*}
  \begin{array}{cll}
    \frac{d\tau}{d\zeta}\Big|_{\zeta=i\omega_{0}} & = & \frac{3\zeta^{2}+2p_{2}\zeta+p_{1}+(2q_{2}\zeta+q_{1})e^{-\tau\zeta}-\tau(q_{2}\zeta^{2}+q_{1}\zeta+q_{0})e^{-\tau\zeta}}{\zeta(q_{2}\zeta^{2} +q_{1}\zeta+q_{0})e^{-\tau\zeta}}\Big|_{\zeta=i\omega_{0}} \\
     & = & \left(-\frac{3\zeta^{2}+2p_{2}\zeta+p_{1}}{\zeta(\zeta^{3}+p_{2}\zeta^{2}+p_{1}\zeta+p_{0})}+\frac{2q_{2}\zeta+q_{1}}{\zeta(q_{2}\zeta^{2} +q_{1}\zeta+q_{0})}-\frac{\tau}{\zeta}\right)\Big|_{\zeta=i\omega_{0}} \\
     & = &-\frac{-3\omega_{0}^{2}+i2p_{2}\omega_{0}+p_{1}}{i\omega_{0}(-i\omega_{0}^{3}-p_{2}\omega_{0}^{2}+ip_{1}\omega_{0}+p_{0})}
     +\frac{i2q_{2}\omega_{0}+q_{1}}{i\omega_{0}(-q_{2}\omega_{0}^{2}+iq_{1}\omega_{0}+q_{0})}-\frac{\tau}{i\omega_{0}}\\
     & = &-\frac{1}{\omega_{0}}\frac{(p_{1}-3\omega_{0}^{2})+i2p_{2}\omega_{0}}{\omega_{0}(\omega_{0}^{2}-p_{1})+i(p_{0}-p_{2}\omega_{0}^{2})}
     -\frac{1}{\omega_{0}}\frac{q_{1}+i2q_{2}\omega_{0}}{q_{1}\omega_{0}-i(q_{0}-q_{2}\omega_{0}^{2})}-\frac{\tau}{i\omega_{0}}\\
     & = & -\frac{1}{\omega_{0}}\frac{[(p_{1}-3\omega_{0}^{2})+i2p_{2}\omega_{0}][\omega_{0}(\omega_{0}^{2}-p_{1})-i(p_{0}-p_{2}\omega_{0}^{2})]}{[\omega_{0}(\omega_{0}^{2}-p_{1})]^{2}+(p_{0}-p_{2}\omega_{0}^{2})^{2}}
     -\frac{1}{\omega_{0}}\frac{(q_{1}+i2q_{2}\omega_{0})[q_{1}\omega_{0}+i(q_{0}-q_{2}\omega_{0}^{2})]}{(q_{1}\omega_{0})^{2}+(q_{0}-q_{2}\omega_{0}^{2})^{2}}
     +\frac{i\tau}{\omega_{0}}.\\
  \end{array}
\end{equation*}
Consequently, we can further get
\begin{equation*}
     \begin{array}{ccl}
       $\mbox{Re}$\left(\frac{d\tau}{d\zeta}\Big|_{\zeta=i\omega_{0}} \right) & = & -\frac{(\omega_{0}^{2}-p_{1})(p_{1}-3\omega_{0}^{2})+2p_{2}(p_{0}-p_{2}\omega_{0}^{2})}{[\omega_{0}(\omega_{0}^{2}-p_{1})]^{2}+(p_{0}-p_{2}\omega_{0}^{2})^{2}}
       -\frac{q_{1}^{2}-2q_{2}(q_{0}-q_{2}\omega_{0}^{2})}{(q_{1}\omega_{0})^{2}+(q_{0}-q_{2}\omega_{0}^{2})^{2}} \\
        & = & \frac{-(\omega_{0}^{2}-p_{1})(p_{1}-3\omega_{0}^{2})-2p_{2}(p_{0}-p_{2}\omega_{0}^{2})-q_{1}^{2}+2q_{2}(q_{0}-q_{2}\omega_{0}^{2})}
        {(q_{1}\omega_{0})^{2}+(q_{0}-q_{2}\omega_{0}^{2})^{2}}. \\
     \end{array}
\end{equation*}
Since
\begin{equation*}
  C_{2}=p_{2}^{2}-q_{2}^{2}-2p_{1}>0  \quad\quad\mbox{and}\quad\quad  C_{1}=p_{1}^{2}-q_{1}^{2}-2p_{2}p_{0}+2q_{2}q_{0}>0,
\end{equation*}
we conclude that
\begin{equation*}
  \begin{array}{ccl}
    \mbox{sign}\left(\frac{\mbox{d}\rm{Re}(\zeta)}{\mbox{d}\tau}\Big|_{\tau=\tau_{k}}\right) & = & \mbox{sign}\left(\mbox{Re}\left(\frac{\mbox{d}\tau}{\mbox{d}\zeta}\Big|_{\zeta=i\omega_{0}}\right)\right) \\
     & = & \mbox{sign}\left(\frac{-(\omega_{0}^{2}-p_{1})(p_{1}-3\omega_{0}^{2})-2p_{2}(p_{0}-p_{2}\omega_{0}^{2})-q_{1}^{2}+2q_{2}(q_{0}-q_{2}\omega_{0}^{2})}
        {(q_{0}-q_{2}\omega_{0}^{2})^{2}+(q_{1}\omega_{0})^{2}}\right) \\
     & = & \mbox{sign}\left(\frac{3\omega_{0}^{4}+2(p_{2}^{2}-q_{2}^{2}-2p_{1})\omega_{0}^{2}+p_{1}^{2}-q_{1}^{2}-2p_{2}p_{0}+2q_{2}q_{0}}
        {(q_{0}-q_{2}\omega_{0}^{2})^{2}+(q_{1}\omega_{0})^{2}}\right)>0. \\
  \end{array}
\end{equation*}
\end{proof}
\noindent
Summarizing the results presented above, we derive the following theorem.
\begin{theorem}\label{HopfBifurcation}
Let Assumption \ref{assumption1} and \ref{assumption2} be satisfied. Then there exist $\tau_{k}>0, k=1,2,\cdots$($\tau_{k}$ is defined in (\ref{tauk})), such that when $\tau=\tau_{k}$, the HIV model (\ref{system}) undergoes a Hopf bifurcation at the equilibrium $\left(\overline{T}, \overline{V}, \overline{i}_{\tau_{k}}(a)\right)$. In particular, a non-trivial periodic solution bifurcates from the equilibrium $\left(\overline{T}, \overline{V}, \overline{i}_{\tau_{k}}(a)\right)$ when $\tau=\tau_{k}$.
\end{theorem}

\section{Numerical simulations and conclusions}

\noindent

In this section, we perform a numerical analysis of the model (\ref{system}) based on the previous results. We choose a set of parameters as follows: $\Lambda=0.05, r=0.95, \mu=0.0002, \sigma=0.09, c=0.4, K=50, N=30, \beta_{1}=0.00027,\beta_{2}=0.027$. System (\ref{system}) becomes
\begin{figure}[hbt]
\centering
{
\subfigure [] {\includegraphics[width=2in,clip]{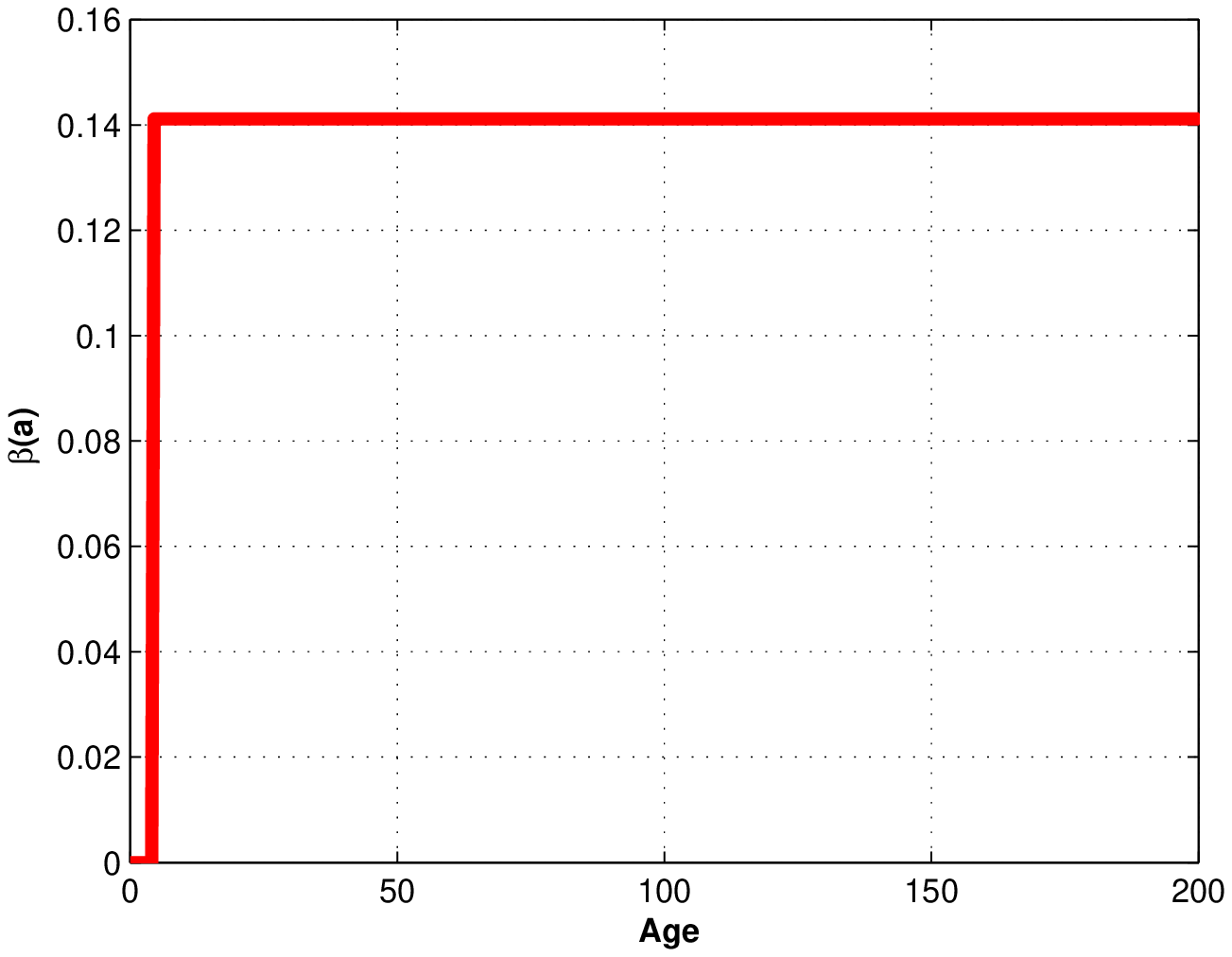} }
\subfigure [] {\includegraphics[width=2in,clip]{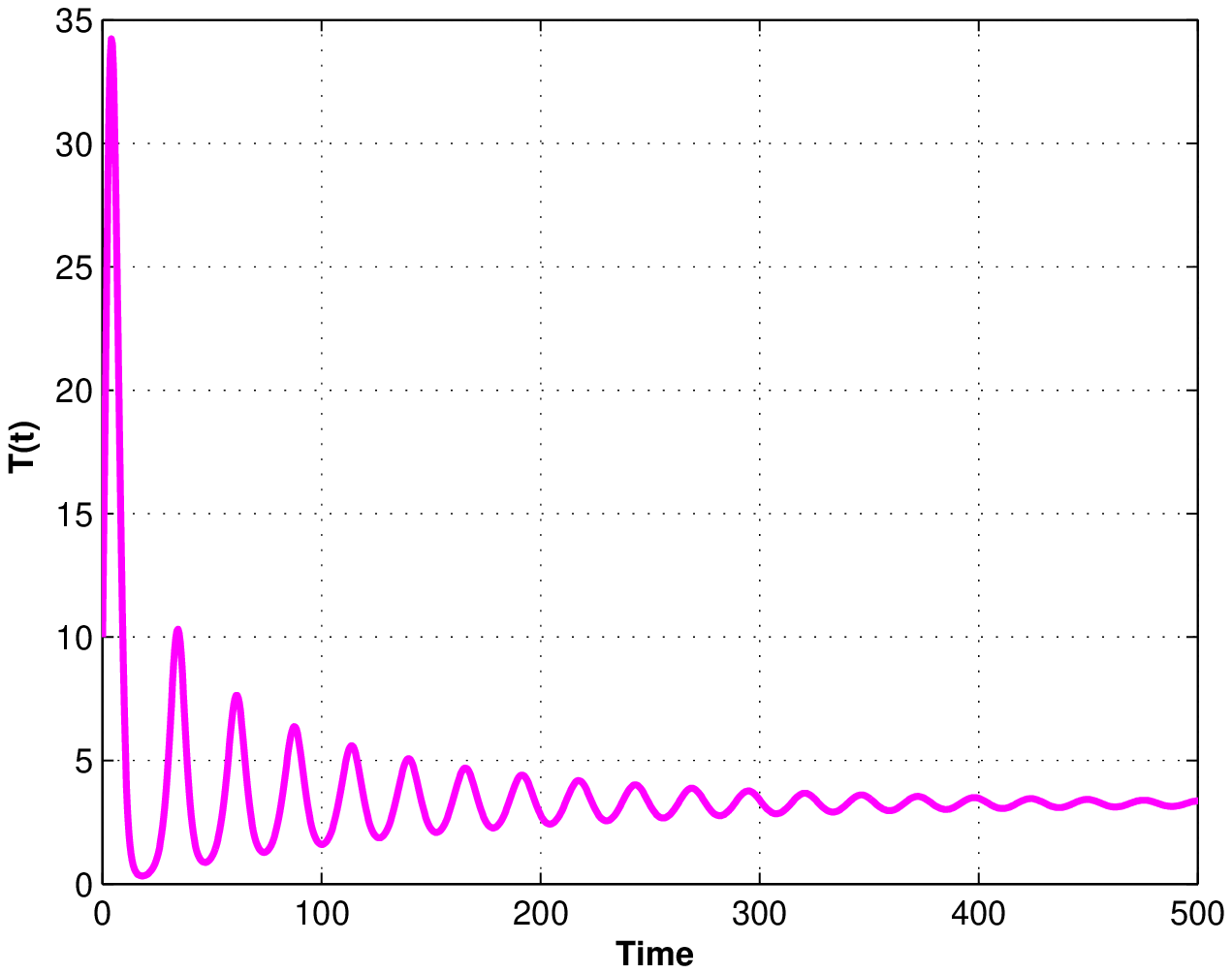} }
\subfigure [] {\includegraphics[width=2in,clip]{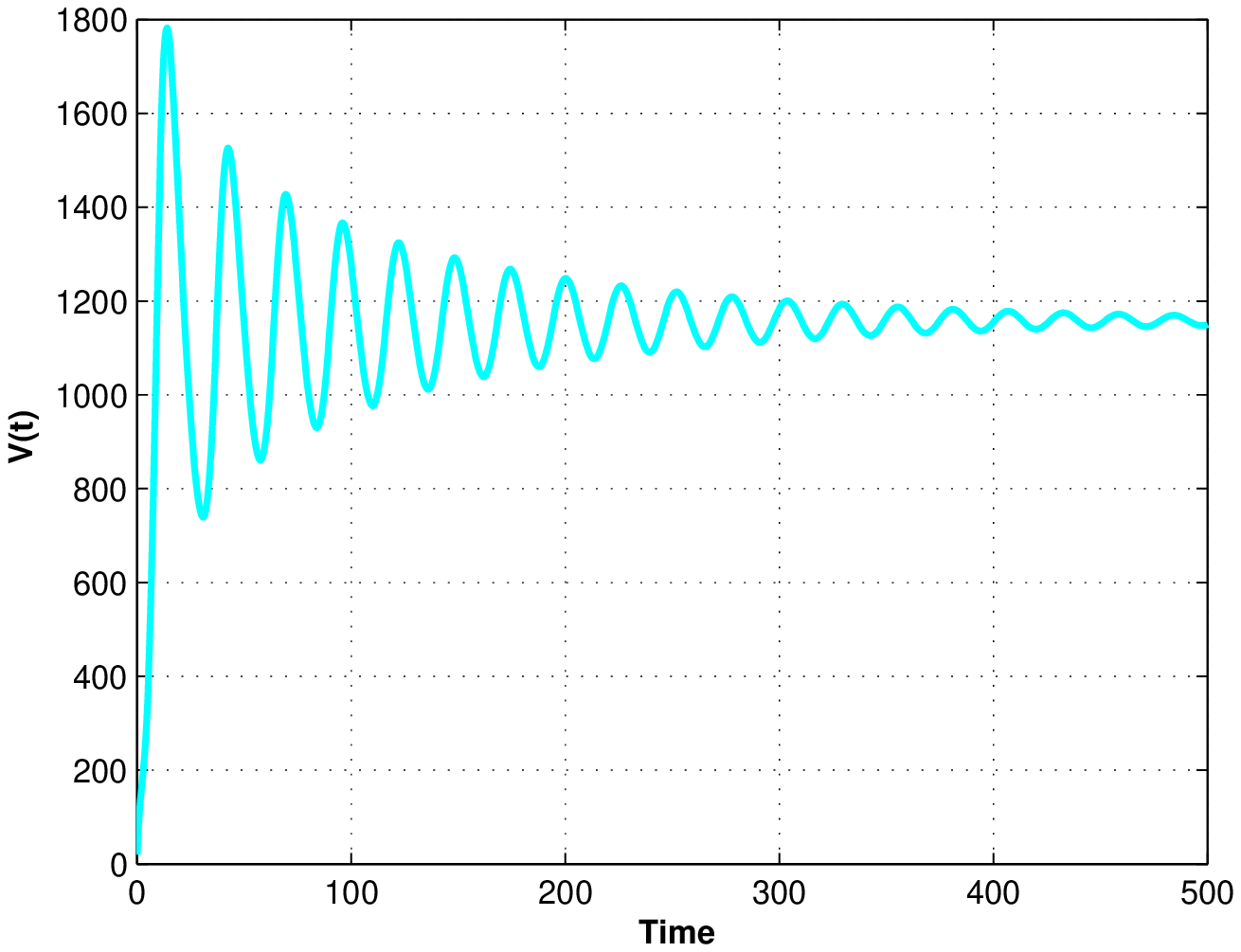} }
\subfigure [] {\includegraphics[width=2in,clip]{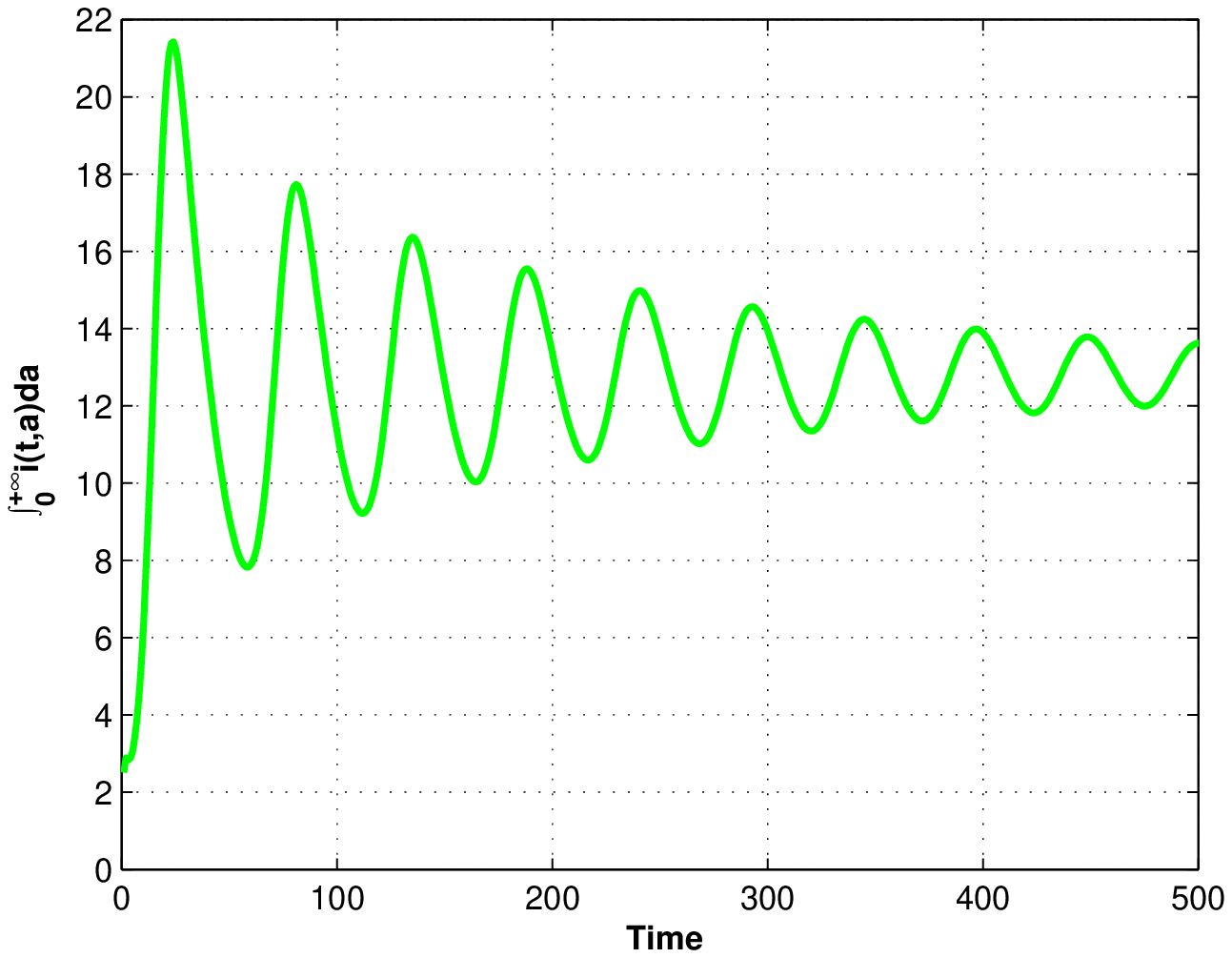} }
\subfigure [] {\includegraphics[width=2in,clip]{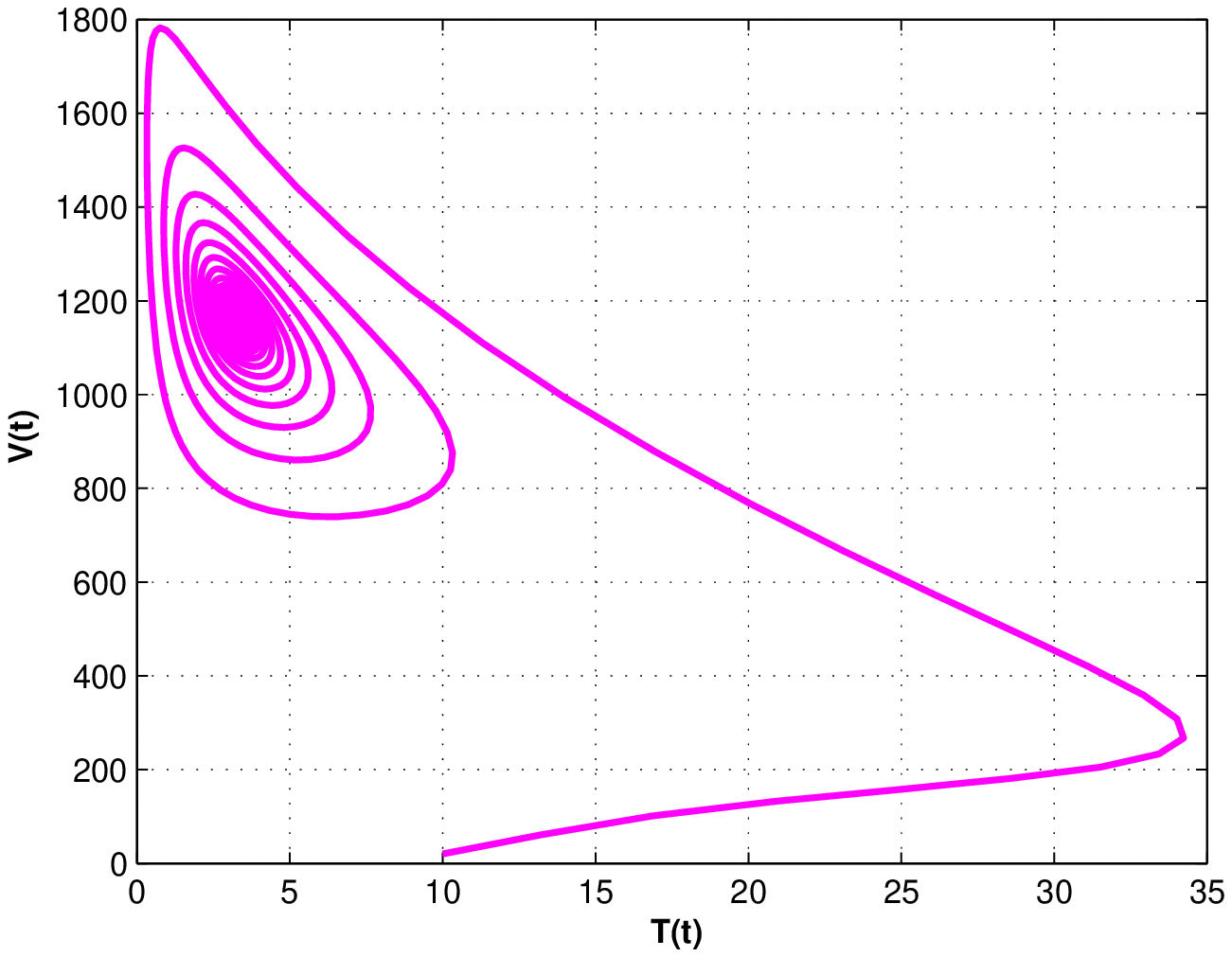} }
\subfigure [] {\includegraphics[width=2in,clip]{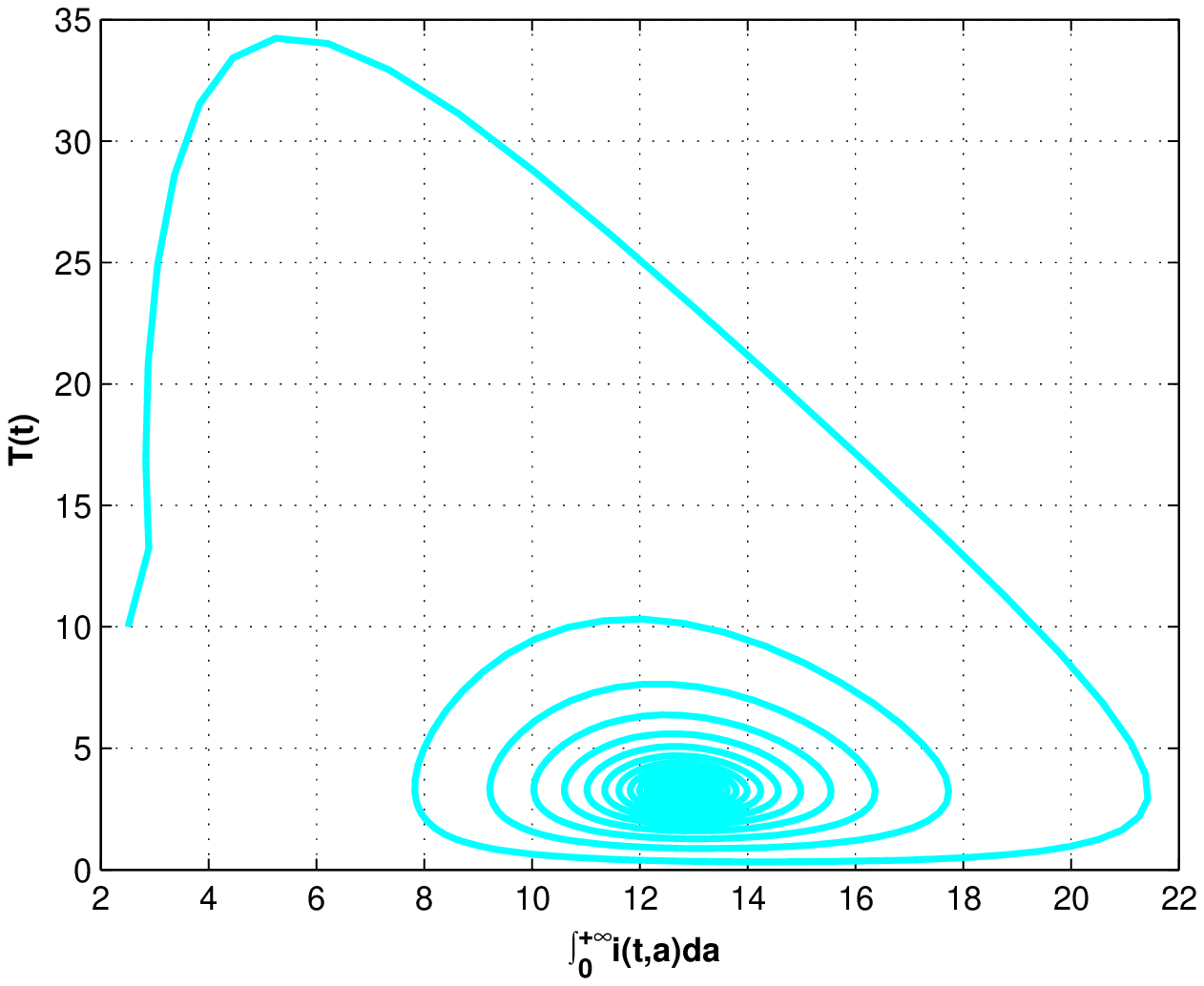} }
\subfigure [] {\includegraphics[width=2in,clip]{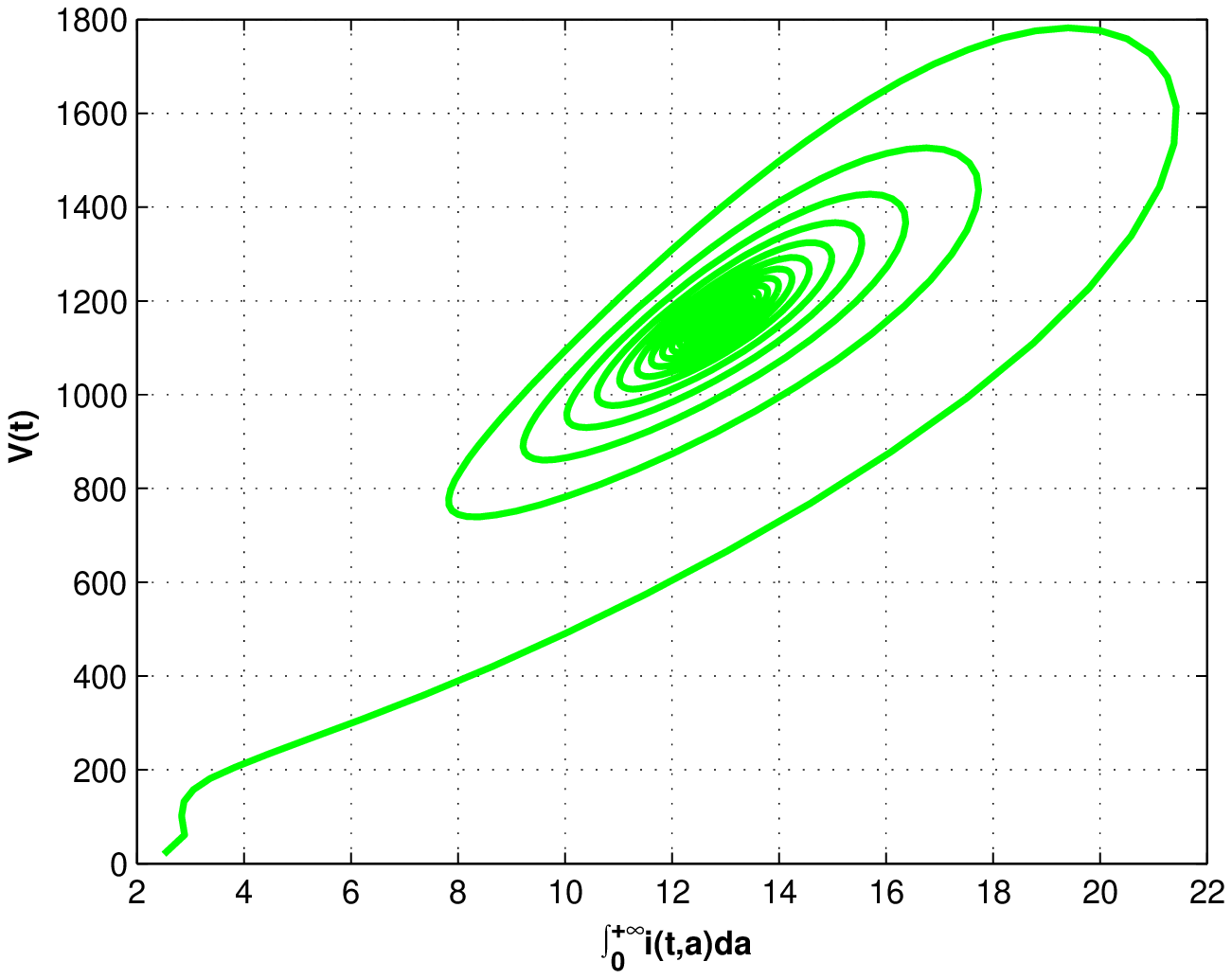} }
\subfigure [] {\includegraphics[width=2in,clip]{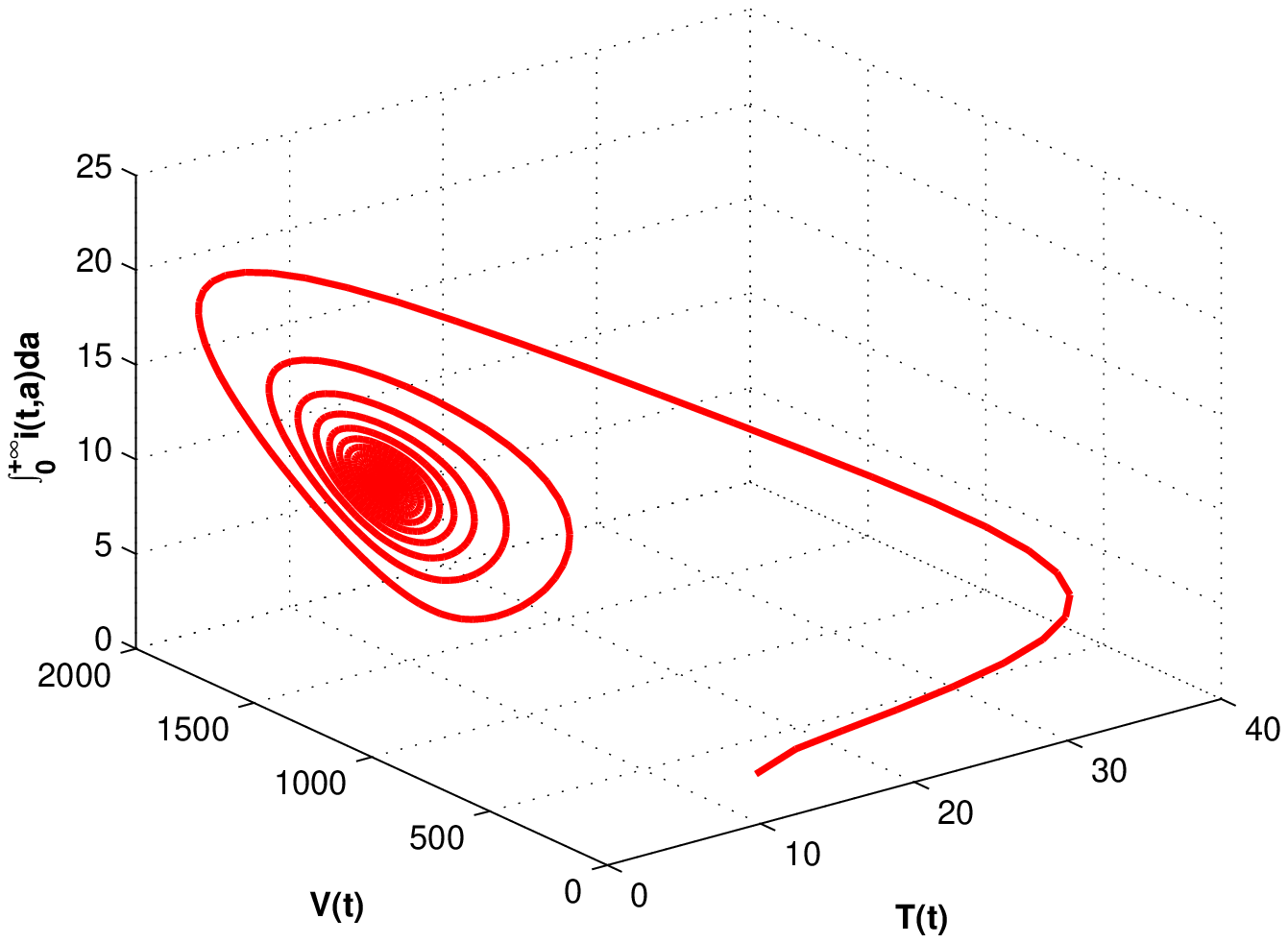} }
\subfigure [] {\includegraphics[width=2in,clip]{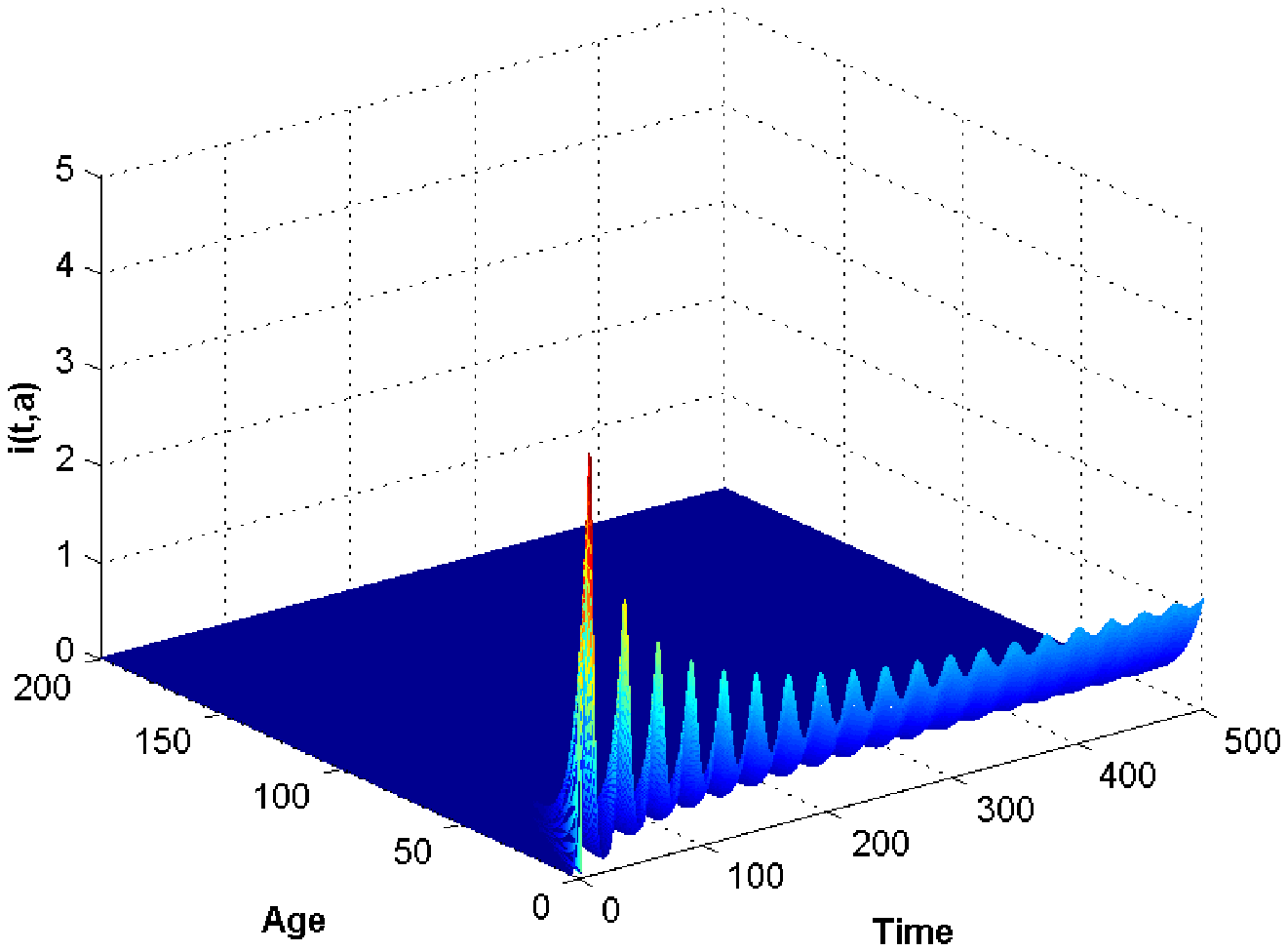} }
}
\caption{\footnotesize{Description of the evolution of system (\ref{system}) when $\tau=5$.}}\label{solution5}
\end{figure}

\begin{figure}[ht]
\centering
{
\subfigure [] {\includegraphics[width=2in,clip]{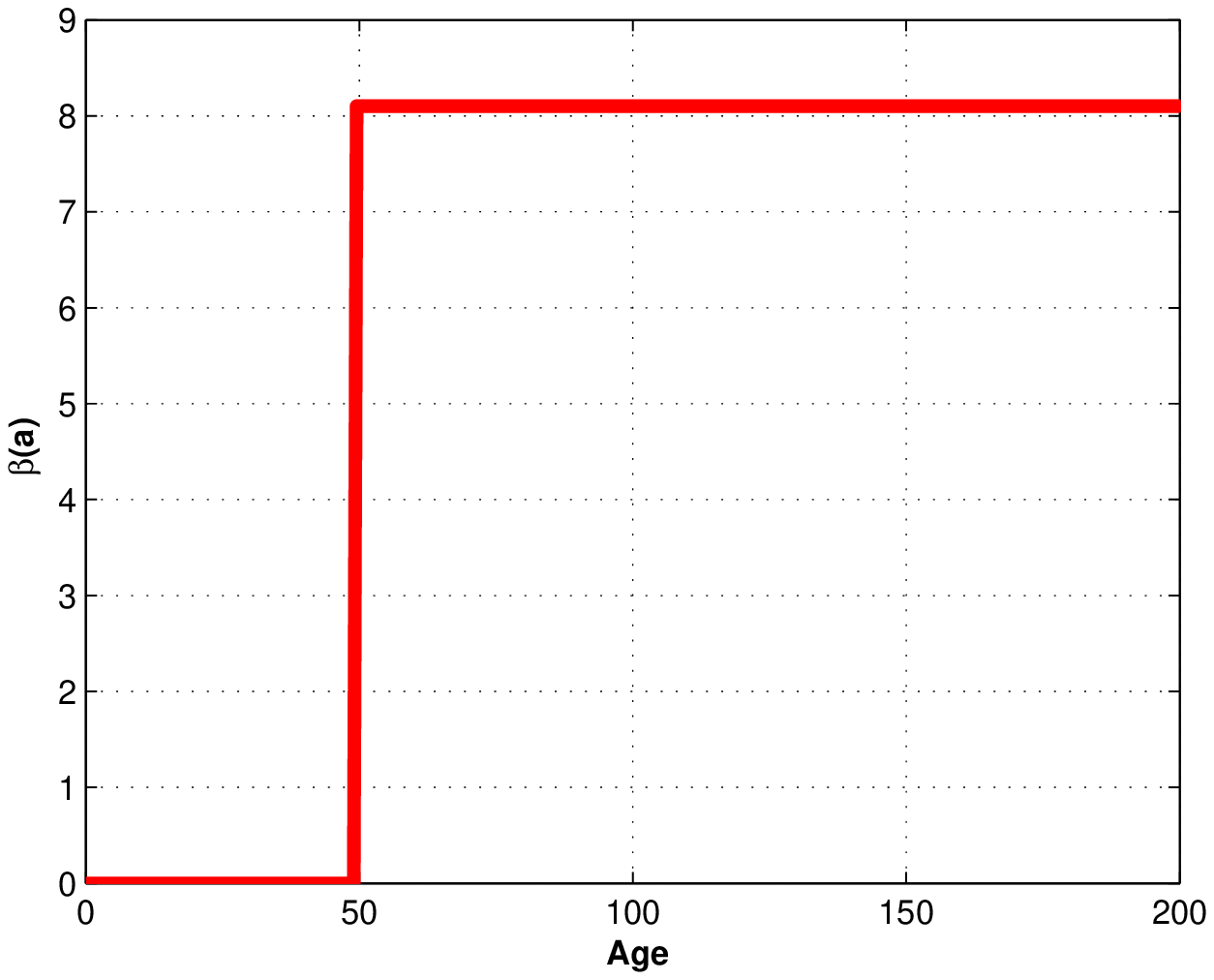} }
\subfigure [] {\includegraphics[width=2in,clip]{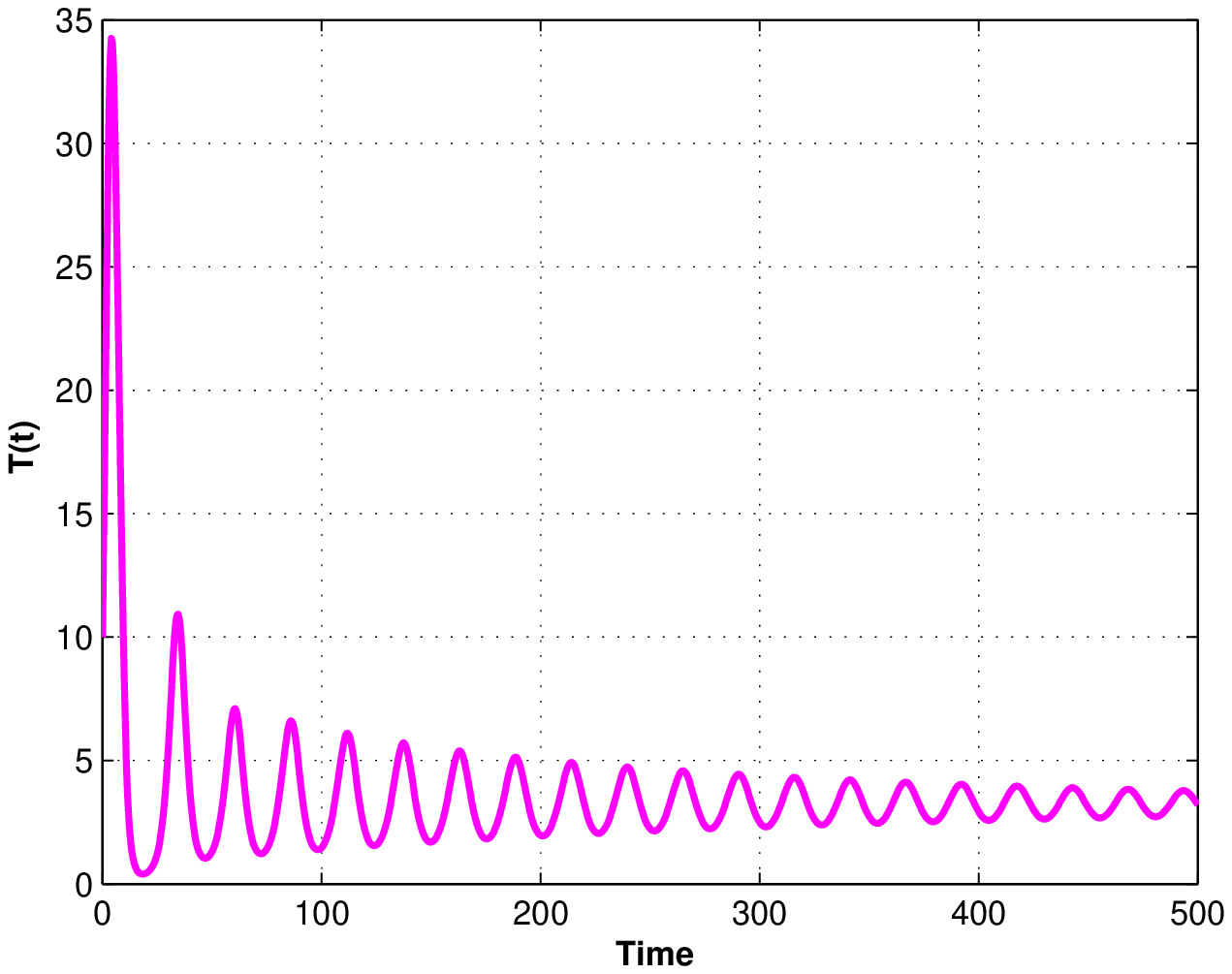} }
\subfigure [] {\includegraphics[width=2in,clip]{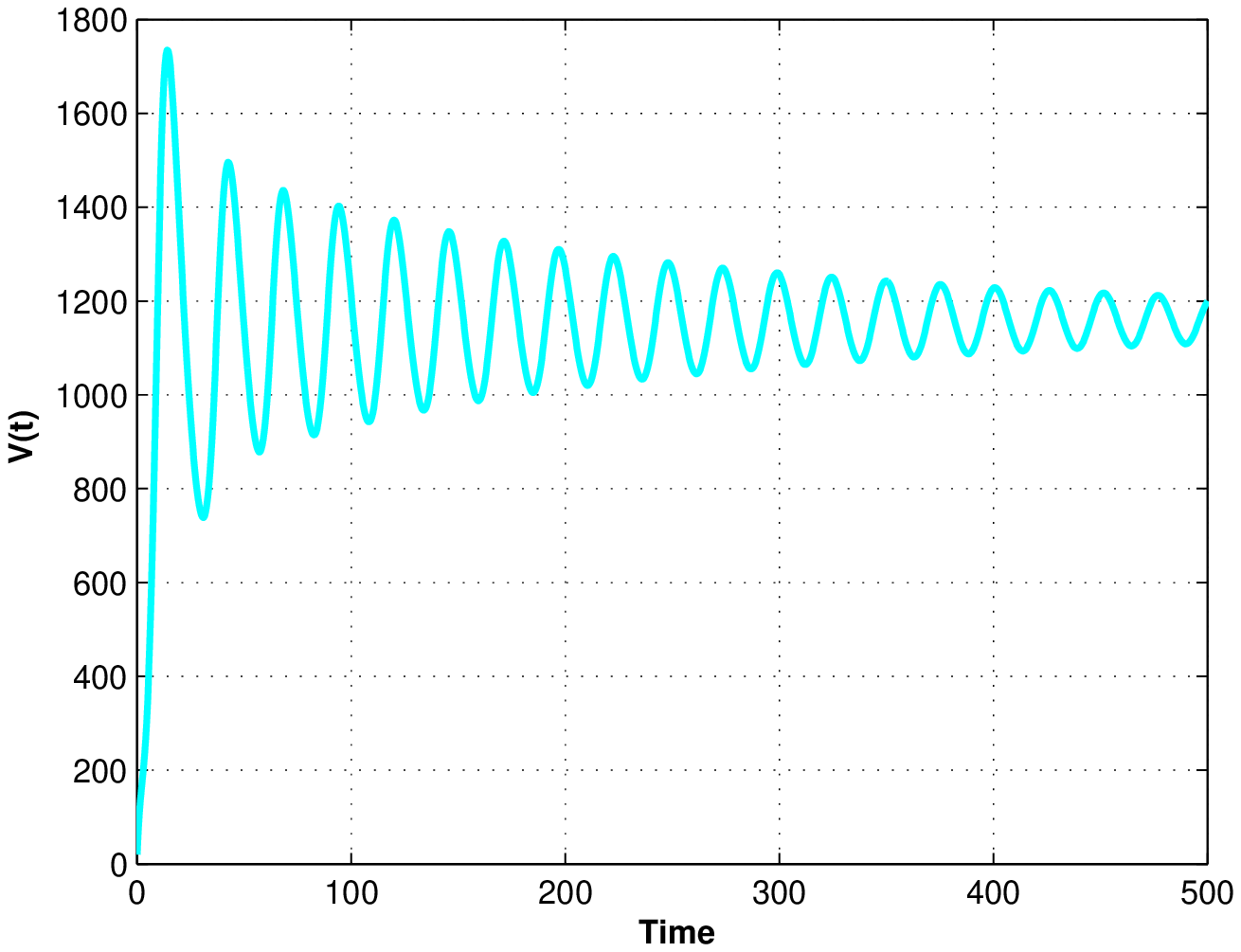} }
\subfigure [] {\includegraphics[width=2in,clip]{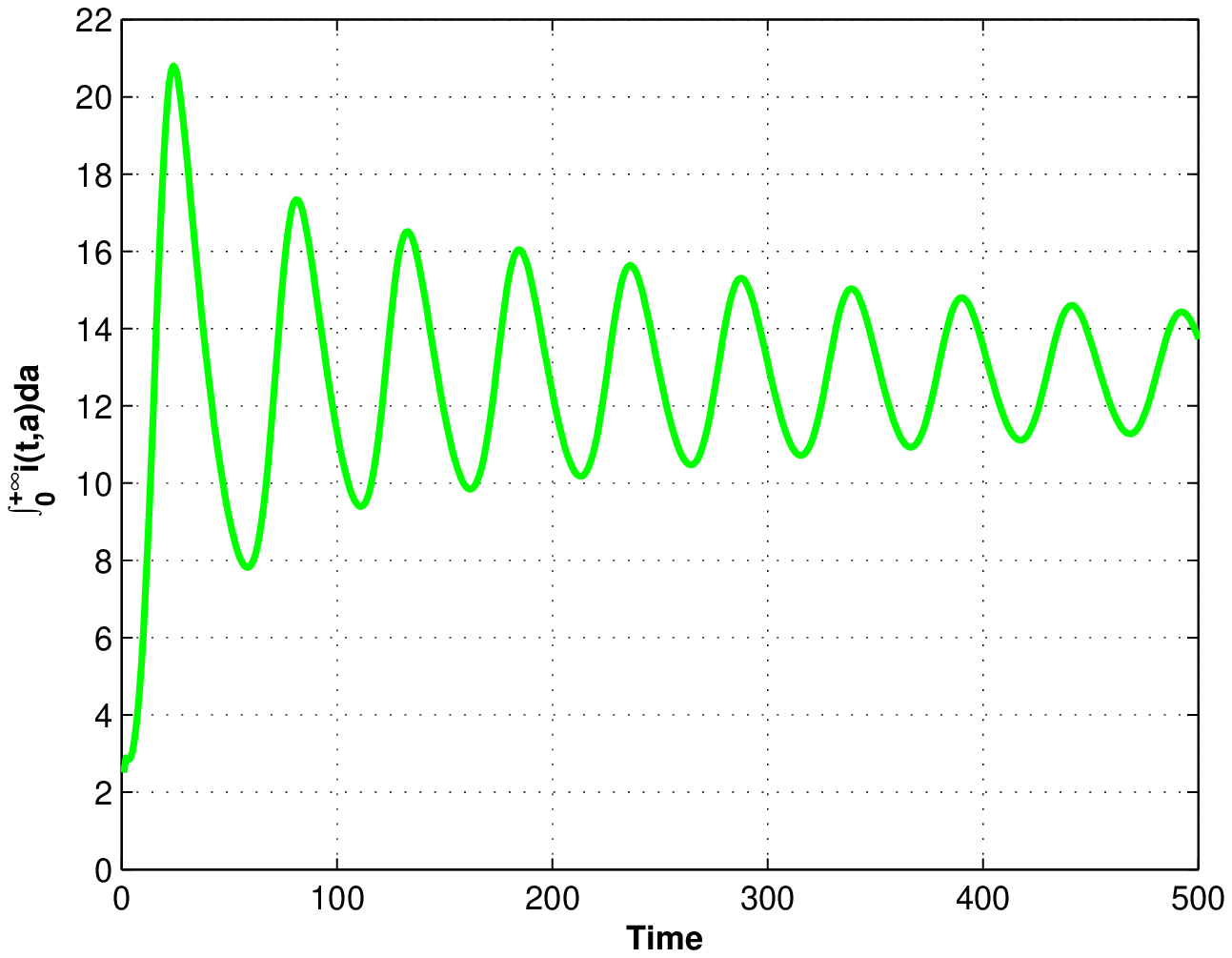} }
\subfigure [] {\includegraphics[width=2in,clip]{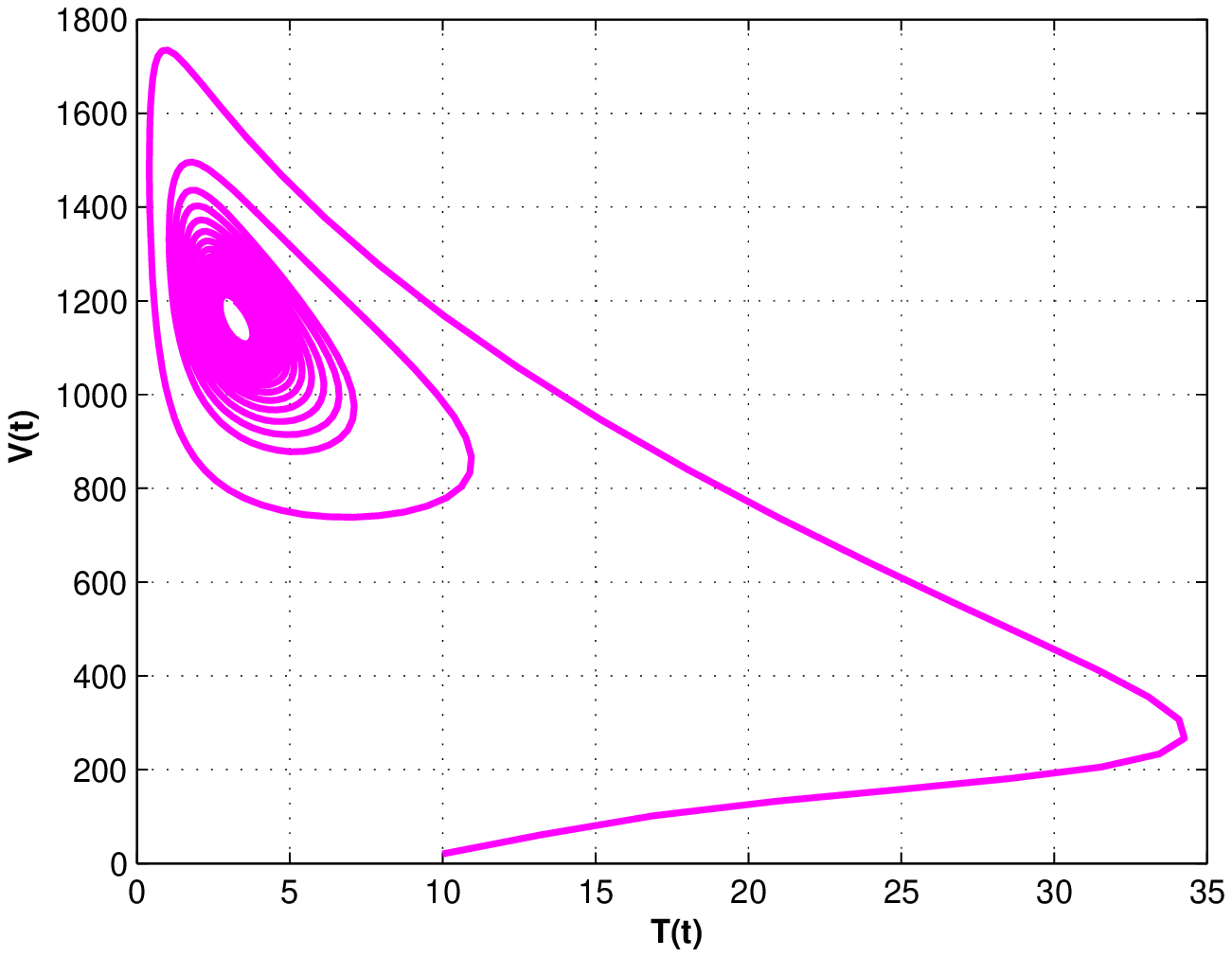} }
\subfigure [] {\includegraphics[width=2in,clip]{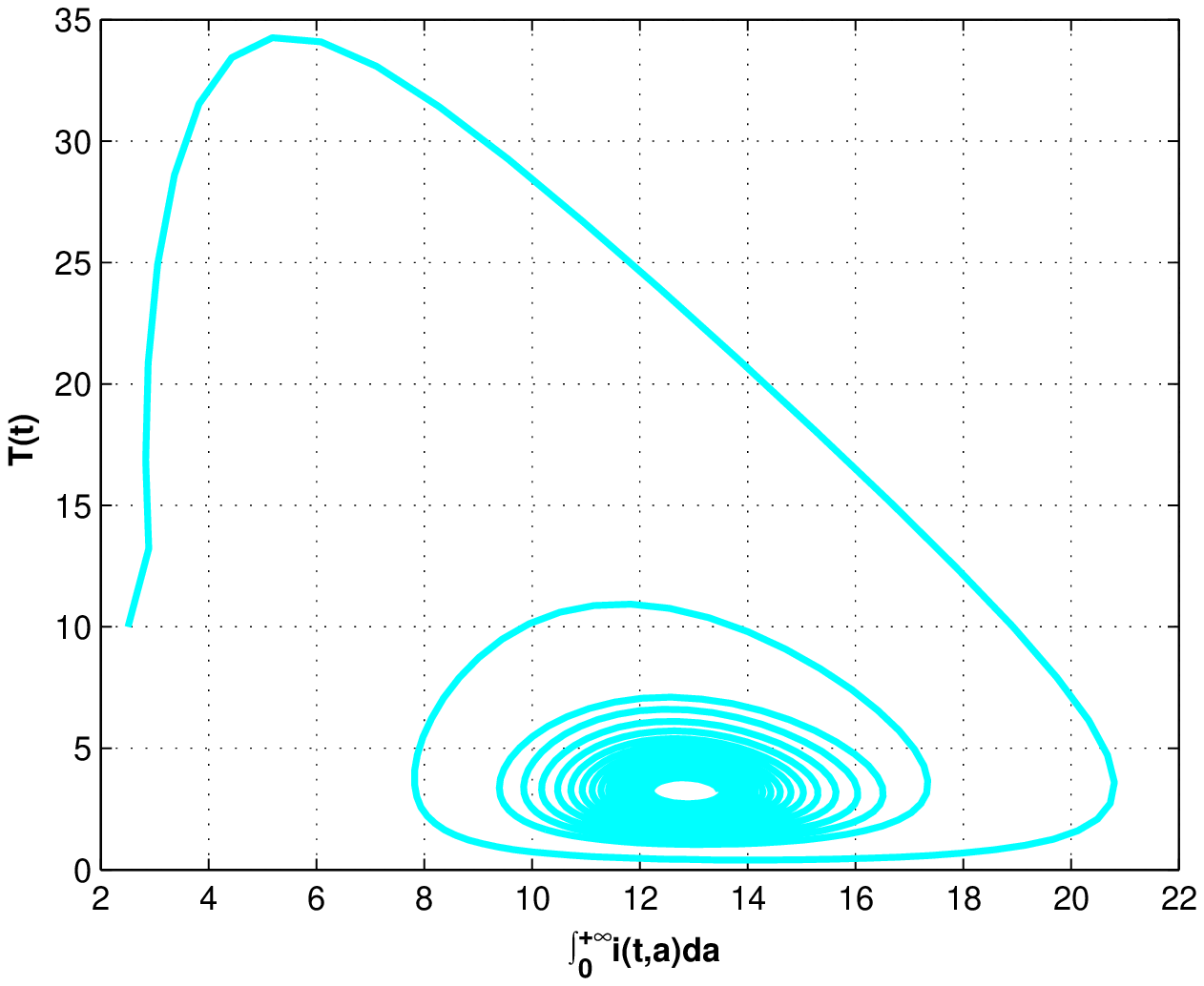} }
\subfigure [] {\includegraphics[width=2in,clip]{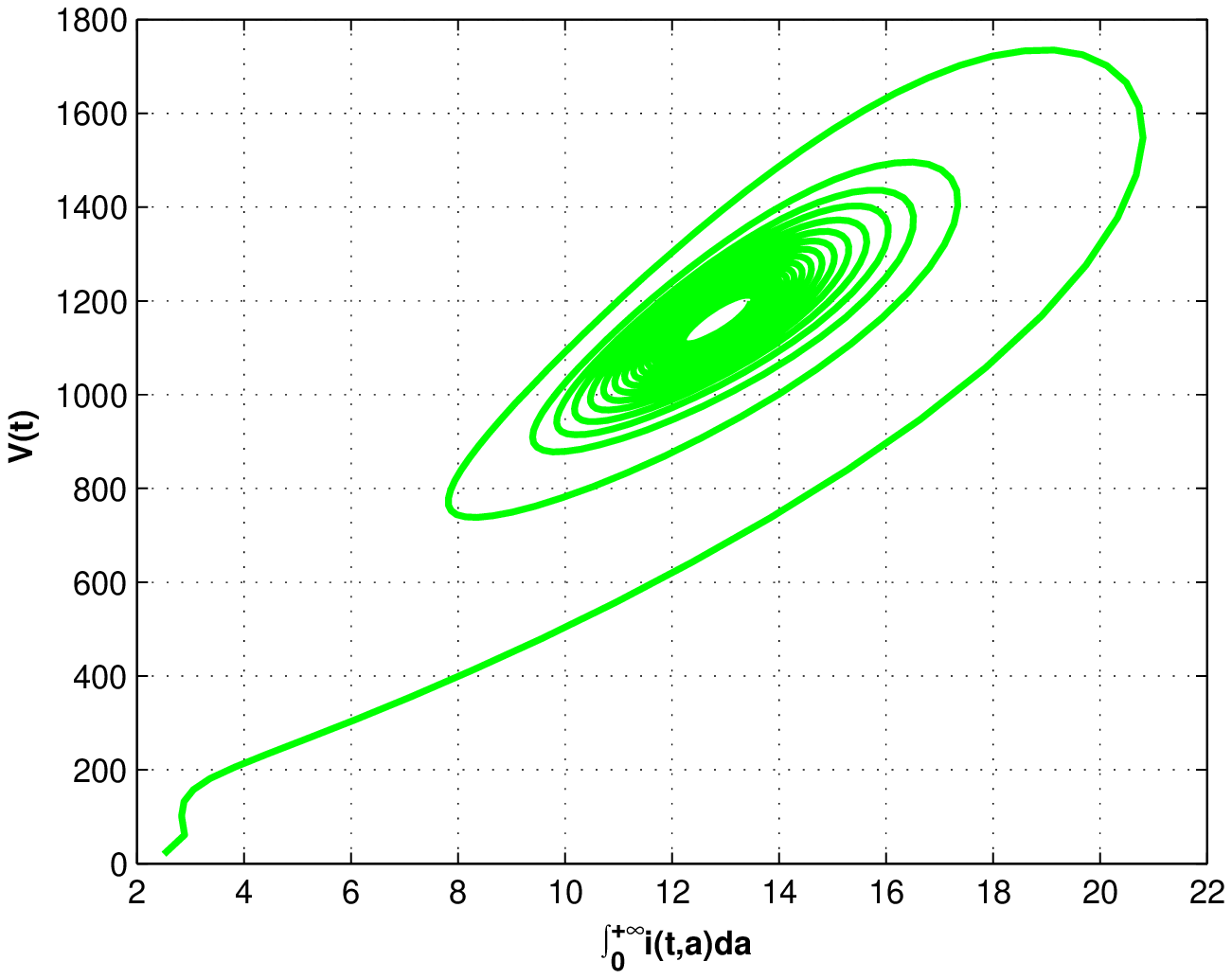} }
\subfigure [] {\includegraphics[width=2in,clip]{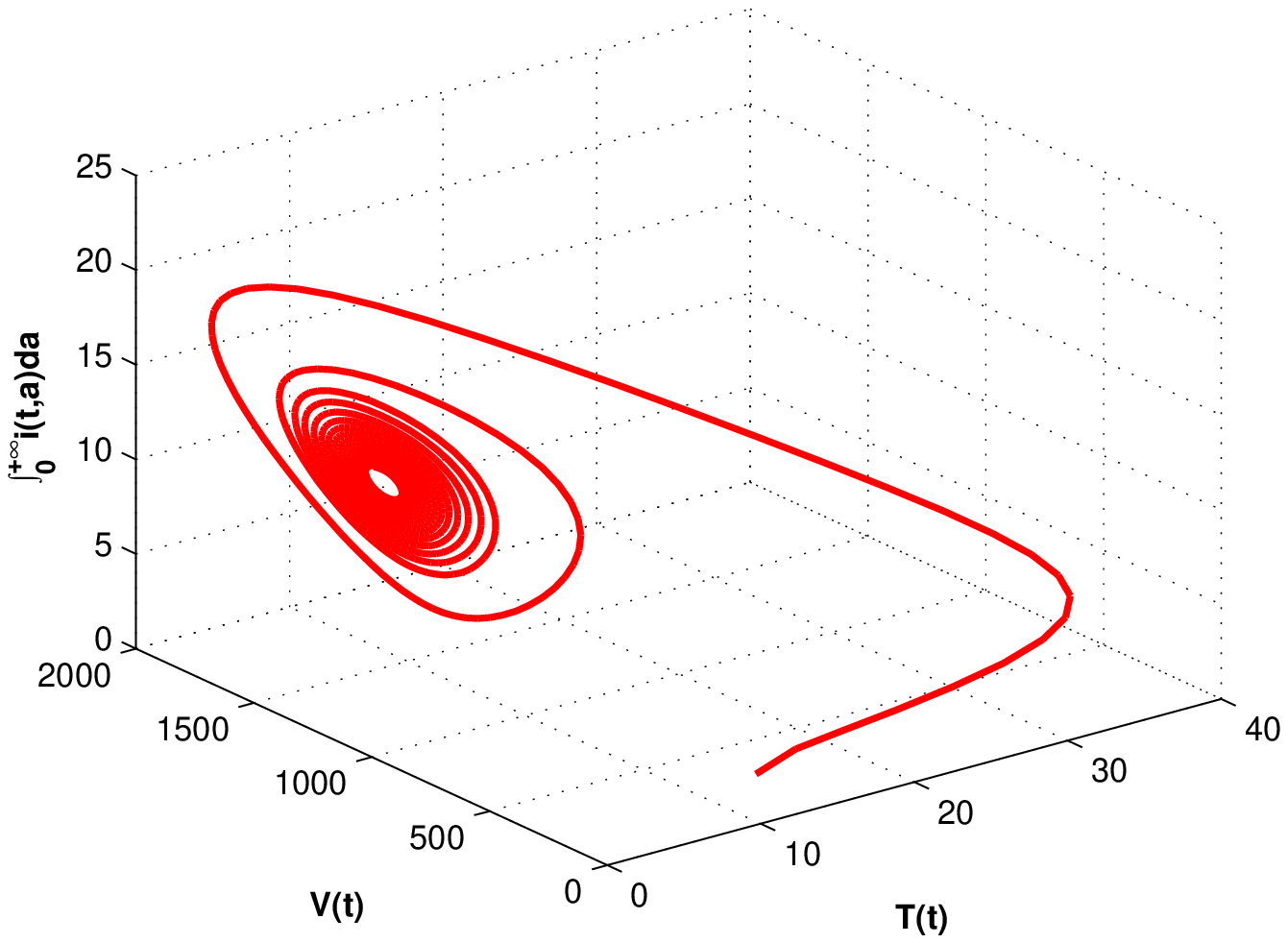} }
\subfigure [] {\includegraphics[width=2in,clip]{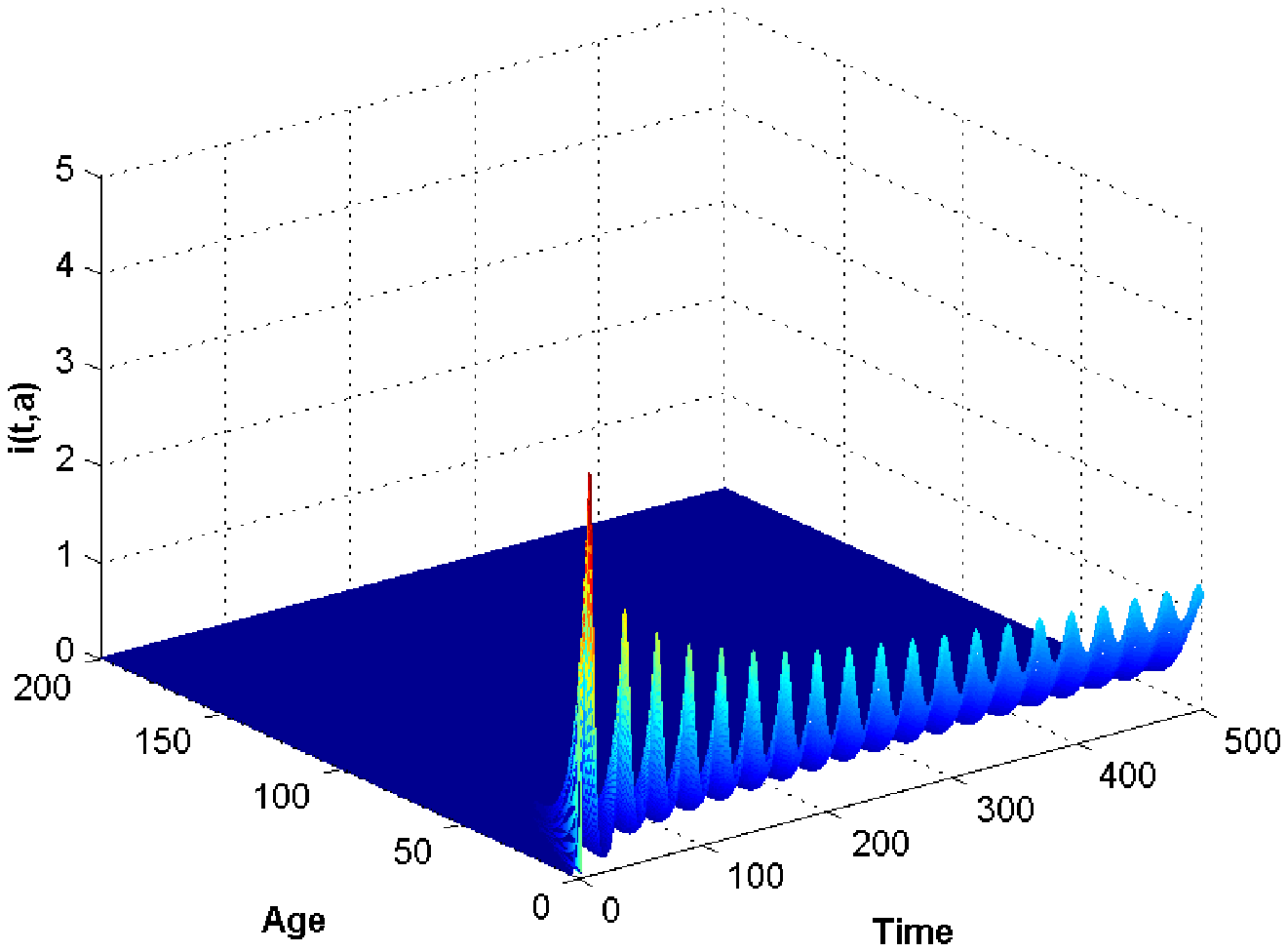} }
}
\caption{\footnotesize{Description of the evolution of system (\ref{system}) when $\tau=50$.}}\label{solution50}
\end{figure}

\begin{equation}
\left\{
\begin{array}{l}
\frac{dT(t)}{dt}=0.05-0.0002 T(t)+0.95T(t)\left(1-\frac{T(t)+\int_{0}^{+\infty}{i(t,a)da}}{50}\right)-0.00027 T(t)V(t)\\
\quad\quad\quad-0.027 T(t)\int_{0}^{+\infty}{i(t,a)da},\\
\frac{dV(t)}{dt}=30\cdot0.09\int_{0}^{+\infty}{i(t,a)da}-0.4V(t),\\
\frac{\partial i(t,a)}{\partial t}+\frac{\partial i(t,a)}{\partial a} = -0.09 i(t,a),\\
i(t,0)=0.00027 T(t)V(t)+0.027 T(t)\int_{0}^{+\infty}{\beta(a)i(t,a)da}, t>0,\\
T(0)=10\geq0, V(0)=20\geq0, i(0,\cdot)=2e^{-\frac{a^{2}}{2}}\in L_{+} ^{1}((0, + \infty ),\mathbb{R}),
\end{array}
\right.
\end{equation}
where
\begin{equation*}
  \beta(a):=\left\{
               \begin{array}{cl}
                 0.09 e^{0.09\tau}, &\quad \mbox{if} \quad a\geq \tau, \\
                 0, & \quad \mbox{if} \quad a\in (0,\tau). \\
               \end{array}
             \right.
\end{equation*}

By using the Matlab, we calculate that $K(N\beta_{1}+c\beta_{2})[\Lambda(N\beta_{1}+c\beta_{2})+c(r-\mu)]-c^{2}r\approx0.2079$, $D\approx2.8735\times 10^{-8}, q\approx1.7436\times 10^{-4}$, $C_{2}\approx0.2647, C_{1}\approx0.0198$. It is obvious that the conditions of Assumption \ref{assumption2} can be satisfied. Calculating it further, we can easily obtain that $\omega_{0}\approx0.2364$ and critical value $\tau_{1}\approx9.3906$.

For the above parameters, we draw the graph of the solution curve and phase trajectory of model (\ref{system}) and the graph of $i(t,a)$ with respect to age and time (horizontal axis) by software Matlab when $\tau=5<\tau_{1}$ (Figure \ref{solution5}) and $\tau=50>\tau_{1}$ (Figure \ref{solution50}). One can see that positive equilibrium $\left(\overline{T}, \overline{V}, \overline{i}_{\tau=5}(a)\right)=\left(21.1640, 77.6373,  5.1758e^{-0.45a}\right)$ is locally asymptotically stable when $\tau=5\in[0,\tau_{1})$.
By using Theorem \ref{HopfBifurcation}, we know that, under the set parameters, when $\tau=\tau_{1}$, the HIV model (\ref{system}) undergoes a Hopf bifurcation at the equilibrium $\left(\overline{T}, \overline{V}, \overline{i}_{\tau_{1}}(a)\right)$. As is shown in Figure \ref{solution50}, when bifurcation parameter $\tau $
crosses the bifurcation critical value $\tau _{1}$, the sustained periodic
oscillation phenomenon appears around the positive equilibrium $\left(\overline{T}, \overline{V}, \overline{i}_{\tau=50}(a)\right)=\left(21.1640, 77.6373,  51.7582e^{-4.5a}\right)$. Biologically speaking, for
smaller biological maturation period $\tau $, the stability of the unique positive equilibrium of system (%
\ref{system}) is barely  affected. However, when maturation period $\tau$ increases continuously, the dynamical behavior of system (\ref{system}) will be resulted in substantial changes. In conclusion, the bifurcation parameter $\tau $, a measure of a biological maturation period, has an essential impact on the dynamical behavior of system (\ref{system}).

\end{document}